\documentclass{article}
 
\usepackage[final]{neurips_2024}

\usepackage{natbib}
\setcitestyle{numbers,square,comma}
\usepackage[utf8]{inputenc} % allow utf-8 input
\usepackage[T1]{fontenc}    % use 8-bit T1 fonts
\usepackage{hyperref}       % hyperlinks
\usepackage{url}            % simple URL typesetting
\usepackage{booktabs}       % professional-quality tables
\usepackage{amsfonts}       % blackboard math symbols
\usepackage{nicefrac}       % compact symbols for 1/2, etc.
\usepackage{microtype}      % microtypography
\usepackage{xcolor}         % colors
\usepackage{graphicx}

\usepackage{amsmath}
\usepackage{amssymb}
\usepackage{mathtools}
\usepackage{amsthm}
\usepackage{bm}
\usepackage{bbm}
\usepackage{braket}
\usepackage{wrapfig}
\usepackage{algorithm}
\usepackage{algorithmic}
\usepackage{multirow}
 
\usepackage[capitalize,noabbrev]{cleveref}

\theoremstyle{plain}
\newtheorem{theorem}{Theorem}[section]

\newtheorem{lemma}[theorem]{Lemma}

\theoremstyle{definition}
\newtheorem{definition}[theorem]{Definition}

\theoremstyle{remark}

\usepackage[textsize=tiny]{todonotes}

\DeclareMathOperator{\eff}{eff}
\DeclareMathOperator{\Tr}{Tr}
\DeclareMathOperator{\Per}{Per}

\title{MG-Net: Learn to Customize QAOA with Circuit Depth Awareness}

\author{
Yang Qian$^{1,\dagger}$ \quad Xinbiao Wang$^{2,\ddag}$ \quad Yuxuan Du$^{3,\S}$ \thanks{Corresponding authors}  \quad Yong Luo$^{2,\P}$ \quad Dacheng Tao$^{1,3,\diamondsuit}$ \\
  $^{1}$School of Computer Science, Faculty of Engineering, University of Sydney \\
  New South Wales 2008, Australia \\
  $^{2}$Institute of Artificial Intelligence, School of Computer Science, Wuhan University \\
  Wuhan, China \\
  $^{3}$College of Computing and Data Science, Nanyang Technological University \\
  Singapore 639798, Singapore \\
  $^{\dagger}$\texttt{qianyang1217@gmail.com} \quad $^{\ddag}$\texttt{cyriewang@gmail.com} \quad $^{\S}$\texttt{duyuxuan123@gmail.com} \\
  $^{\P}$\texttt{luoyong@whu.edu.cn} \quad $^{\diamondsuit}$\texttt{dacheng.tao@ntu.edu.sg} 
}

\begin{document}

\maketitle

\begin{abstract}
  Quantum Approximate Optimization Algorithm (QAOA) and its variants exhibit immense potential in tackling combinatorial optimization challenges. However, their practical realization confronts a dilemma: the requisite circuit depth for satisfactory performance is problem-specific and often exceeds the maximum capability of current quantum devices. To address this dilemma, here we first analyze the convergence behavior of QAOA, uncovering the origins of this dilemma and elucidating the intricate relationship between the employed mixer Hamiltonian, the specific problem at hand, and the permissible maximum circuit depth. Harnessing this understanding, we introduce the Mixer Generator Network (MG-Net), a unified deep learning framework adept at dynamically formulating optimal mixer Hamiltonians tailored to distinct tasks and circuit depths. Systematic simulations, encompassing Ising models and weighted Max-Cut instances with up to 64 qubits, substantiate our theoretical findings, highlighting MG-Net's superior performance in terms of both approximation ratio and efficiency.
\end{abstract}

\section{Introduction}\label{sec:intro}

Combinatorial optimization problems (COPs) \cite{ausiello2012complexity}, central to numerous scientific and engineering disciplines \cite{commander2009maximum,jensen2011graph,hoffman2013traveling}, often defy efficient classical solutions due to their computational complexity \cite{papadimitriou1998combinatorial,karp2010reducibility}. A promising strategy to overcome these computational challenges involves harnessing the power of quantum computing, as these COPs can be mapped to Ising Hamiltonians whose ground states denote optimal solutions \cite{lucas2014ising,oh2019solving}. Leveraging this quantum representation, the Quantum Approximate Optimization Algorithm (QAOA) \cite{farhi2014quantum} has emerged to address these COPs. In particular, theoretical analyses  \cite{farhi1602quantum,lloyd2018quantum,morales2020universality,blekos2023review} underscore the potential of QAOA, suggesting its superiority over classical counterparts in certain contexts, particularly with unlimited infinite circuit depth. Meantime, empirical studies \cite{wang2018quantum,pagano2020quantum,zhou2020quantum} affirm its applicability across a diverse spectrum of problems and devices.

\begin{wrapfigure}{R}{0.4\textwidth}
    \begin{center}
     \includegraphics[width=0.4\textwidth]{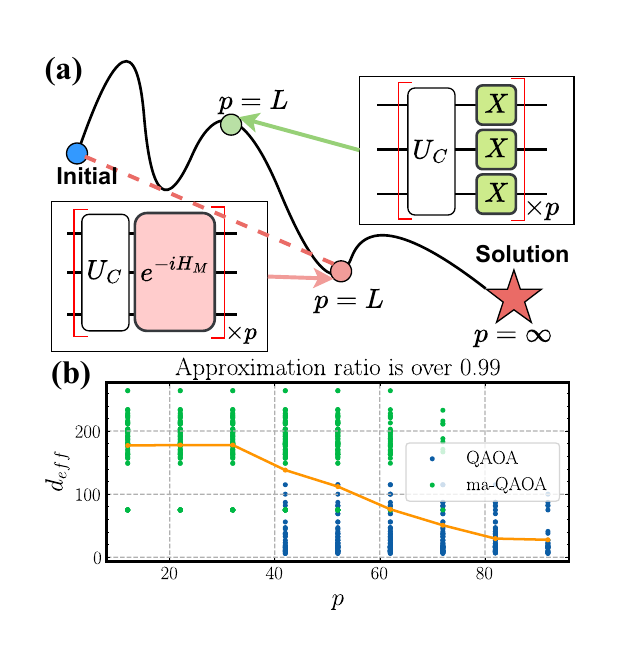}
        \caption{\small{\textbf{Mixer Hamiltonian affects the performance of QAOA}. (a) The optimization trajectories of QAOA with varied mixer Hamiltonians $H_M$. Given a fixed circuit depth $p$, a tailored $H_M$ (highlighted in pink) can more effectively steer the quantum state towards the exact solution compared to the original $H_M$ used in QAOA. (b) Transition of the effective dimension $d_{eff}$ required in QAOA with increasing $p$. `ma-QAOA' denotes a case with independent parameters \cite{herrman2022multi}, contrasted with `QAOA' where parameters are fully correlated. The orange line denotes the average effective dimension over all samples.}}
        \label{fig:opt-traj}
    \end{center}
\end{wrapfigure}

Despite these advancements, QAOA's practical efficacy is challenged by the quantum coherence limits of modern quantum devices, as there is a ceiling on the allowable maximum circuit depth $p$. As a result, standard QAOA often underperforms classical counterparts \cite{moll2018quantum,guerreschi2019qaoa}. This motivates a research shift towards redesigning the \textit{mixer Hamiltonian} $H_M$, a key component of QAOA. As illustrated in Fig.~\ref{fig:opt-traj}(a), supported by the results of quantum adiabatic evolution \cite{berry2009transitionless,guery2019shortcuts},  alternative $H_M$ may exist that guide the system along a more direct and efficient trajectory—a shortcut—to the solution state, leading to a better performance compared to the standard QAOA.  Besides, as shown in Fig.~\ref{fig:opt-traj}(b), empirical evidence indicates that the form of $H_M$ promising a good performance is varied with the allowable $p$. As such, diverse alternatives $H_M$ are proposed in past years, drawing upon concepts from quantum annealing \cite{yu2022quantum}, incorporating additional trainable parameters \cite{herrman2022multi} or exploiting permutation symmetry \cite{sauvage2022building}. However, these approaches require deep domain expertise and often lack generalizability across different tasks and circuit configurations $p$. 

In response to these challenges, here we first analyze the convergence of QAOA on various mixer Hamiltonian configurations and circuit depths with the tool of representation theory \cite{williams2002representation}. Our finding reveals that (i) the convergence of QAOA can be enhanced through parameter grouping in the mixer Hamiltonian; (ii) the specific strategy for parameter grouping is dependent on the particular problem and the value of $p$. These two findings are instrumental in understanding the interplay between $p$, parameter grouping, and the overall efficiency of the QAOA, providing valuable insights for the design of the mixer Hamiltonian.

Envisioned by the achieved theoretical results, we propose an end-to-end learning framework, termed \textbf{M}ixer \textbf{G}enerator \textbf{Net}work (MG-Net), to dynamically design the mixer Hamiltonian $H_M$ for a class of problems and distinct circuit depth constraints. Conceptually, MG-Net takes the problem's description and the available circuit depth $p$ as input and directly outputs the optimal mixer Hamiltonian for a $p$-QAOA. There are three distinguished features of our proposal: (i) The ability to dynamically adjust $H_M$ according to $p$, enhancing its \textit{compatibility} with practical quantum devices; (ii) Fast customization of $H_M$ for \textit{unseen problems} and circuit depth $p$, attributed to the multi-condition controlled generative network architecture; (iii) Circumvent the need for the expensive collection of a vast training dataset of optimal $H_M$ by employing an estimator-generator structure alongside a two-stage training approach. Note that the developed techniques can be flexibly extended to other variational quantum algorithms (VQAs) \cite{cerezo2021variational,qian2022dilemma}, which may be independent of interests.

The contributions of this paper are:\\
\noindent$\bullet$ We provide a rigorous theoretical analysis on the convergence of QAOA with sufficient circuit depth, elucidating \textbf{the link between the performance and the parameter grouping in QAOA circuits}. This analysis offers guidance on the design of mixer Hamiltonian to achieve a high approximation ratio for a specified circuit depth.\\
\noindent$\bullet$ We propose MG-Net, which dynamically tailors its predicted mixer Hamiltonian $H_M$ to suit the given problem and circuit depth. Our model \textbf{greatly reduces the cost of collecting labeled training data}, attributed to an estimator-generator framework and a two-stage training strategy.\\
\noindent$\bullet$ The proposed MG-Net demonstrates remarkable \textbf{generalization ability} from a limited dataset to a broad spectrum of combinatorial problems, which facilitates rapid and efficient creation of $H_M$ for unseen problems, advancing the \textbf{practical utility} of QAOAs.\\
\noindent$\bullet$ Extensive experiments on the Transverse-field Ising model and Max-Cut up to $64$ qubits \textbf{verify our theoretical discoveries} and \textbf{demonstrate the advantage of MG-Net in achieving higher approximation ratios at various circuit depths} compared to other quantum and traditional methods. The code is released at \url{https://github.com/QQQYang/MG-Net}.

\section{Background}

\subsection{Quantum approximation optimization algorithm}

Considering a COP defined on a set of $N$ binary variables $\bm{z}=z_1\cdots z_N$, where $z_i\in\{\pm 1\}$, our objective is to identify a bit string $\bm{z}$ that maximizes a specific objective function $C(\bm{z}):\{\pm 1\}^N\rightarrow \mathbb{R}_{\geq 0}$. Intuitively, the solution space grows exponentially with $N$, rendering the exact solution to many COPs intractable \cite{ausiello2012complexity}. In practice, an alternative approximation algorithm is selected to seek an approximate solution $\bm{z}$ to achieve a high approximation ratio $r=C(\bm{z})/C_{\max}$, where $C_{\max}=\max_{\bm{z}}C(\bm{z})$.

In response to this inherent complexity, Quantum Approximate Optimization Algorithm (QAOA) \cite{farhi2014quantum} is proposed. In this  framework, the bit string $\bm{z}$ is encoded into a quantum state $\ket{\bm{x}}=\ket{x_1\cdots x_N}$ with $x_i=(1-z_i)/2$, and the objective function $C(\bm{x})$ is encoded into the problem Hamiltonian $H_C\in \mathbb{C}^{2^N\times 2^N}$ so that $H_C\ket{\bm{x}}=C(\bm{x})\ket{\bm{x}}$. Refer to Appendix~\ref{sec:appdix_Opt_QNNs} for the omitted details.

QAOA is a hybrid quantum-classical algorithm that combines a parameterized quantum circuit (PQC) for state evolution and a classical optimizer for parameter updates. For a $p$-layer QAOA circuit shown in Fig.~\ref{fig:opt-traj}(a), the quantum state $\ket{\psi_p}$ is prepared by alternately applying the problem Hamiltonian $H_C$ and the mixer Hamiltonian $H_M=\sum_{i=1}^NX_i$ on the initial state $\ket{\psi_0}$, formulated as
\begin{equation}\label{eq:para_state}
    \ket{\psi_p(\bm{\alpha},\bm{\beta})}=\prod_{k=1}^p e^{-i\beta_k H_M}e^{-i\alpha_k H_C}\ket{\psi_0},
\end{equation}

where $\bm{\alpha}=(\alpha_1,...,\alpha_p)$ and $\bm{\beta}=(\beta_1,...,\beta_p)$ are $2p$ trainable parameters. These parameters are optimized to maximize the expectation value of the problem Hamiltonian $H_C$:
\begin{equation}\label{eqn:QAOA-cost}
    (\bm{\alpha}^*,\bm{\beta}^*)=\arg \max_{\bm{\alpha},\bm{\beta}} F_p(\bm{\alpha},\bm{\beta}),
\end{equation}

where $F_p(\bm{\alpha},\bm{\beta})=\braket{\psi_p(\bm{\alpha},\bm{\beta})|H_C|\psi_p(\bm{\alpha},\bm{\beta})}$ can be estimated by multiple measurements on the quantum system. As $F_p(\bm{\alpha}^*,\bm{\beta}^*)$ approaches the optimal value $C_{\max}$ of the objective function, we can obtain the approximate solution to the combinatorial optimization problem with high probability by measuring the state $\ket{\psi_p(\bm{\alpha}^*,\bm{\beta}^*)}$ in the computational basis. A metric for assessing the performance of QAOA is the approximation ratio $r=F_p(\bm{\alpha}^*,\bm{\beta}^*)/C_{\max}$.

\subsection{Symmetry in QAOA}
\noindent \textbf{Symmetry, ansatz design, and effective dimension}. A symmetry $S$ refers to the unitary operator leaving the operator $H$ invariant such that $S^{\dagger}HS=H$ (or $[S,H]=0$). All symmetries form a group $\mathcal{S}$ where given any two symmetries $S_1, S_2 \in \mathcal{S}$, the compositions $S_1S_2$ and $S_2S_1$ are also symmetries in $\mathcal{S}$. Among various symmetries, the most relevant one to our work is the permutation symmetry $\pi \in \mathcal{S}_N$, with the subscript being the qubit count $N$ and $\mathcal{S}_N$ being the symmetric group. For example, a permutation $\pi$ with $\pi(1)=3, \pi(2)=1, \pi(3)=2$ acting on the state $\ket{\psi_1}\ket{\psi_2}\ket{\psi_3}$ yields $\pi \ket{\psi_1}\ket{\psi_2}\ket{\psi_3}=\ket{\psi_3}\ket{\psi_1}\ket{\psi_2}$. Throughout the whole study, we denote the group of permutation symmetries of the problem Hamiltonian $H_C$ as 
$\Per(H_C)=\{\pi \in \mathcal{S}_N~|~\pi^{\dagger} H_C \pi =H_C \}.$

Consider an $N$-qubit PQC $U(\bm{\theta})=\prod_{j=1}^{p}\prod_{k=1}^{K}e^{-iH_k\bm{\theta}_{jk}} $ with $\bm{\theta} \in \Theta$ and $d=2^N$. We call $U(\bm{\theta})$ a symmetric PQC with respect to the problem Hamiltonian $H_C$ if there exists a symmetry group $\mathcal{S}$ of $H_C$ such that $[U(\bm{\theta}),S]=0$ for any $\bm{\theta}\in \Theta$ and $S \in \mathcal{S}$. This symmetry is determined by the generators of PQCs $\mathcal{A}=\{H_1, \cdots, H_K\}$ which is also called \textit{ansatz design}, as $[U(\bm{\theta}),S]=0$ holds for any $\bm{\theta}\in \Theta$ if and only if $[H_k, U(\bm{\theta})]=0$ for any $k\in[K]$. Such symmetry can be quantified by the \textit{effective dimension} \cite{you2022convergence,wang2022symmetric}.

\begin{definition}[Effective dimension]\label{def:eff_dim}
    Consider an $N$-qubit QAOA instance ($\ket{\psi_0},U(\bm{\theta}),H_C$) where $U(\bm{\theta})$ acts on the vector space $V$. If there exists a direct sum decomposition $V=\oplus_{j=1}^k V_j$ and $V^*\in\{V_j\}_{j=1}^k$ such that   $U(\bm{\theta})\ket{\psi_0}\in V^*$ for any $\bm{\theta}$ and the ground state of the problem Hamiltonian $\ket{\psi^*}$ satisfies $\ket{\psi^*}\in V^*$, then the effective dimension $d_{\eff}\leq 2^{N}$ is defined as the dimension of $V^*$.
\end{definition}
Experimental and theoretical analysis has shown that symmetric ansatz design with a small effective dimension contributes to better trainability \cite{larocca2022diagnosing,wang2022symmetric,larocca2023theory}. 

\noindent \textbf{Symmetry and ansatz designs in QAOA}. The PQC in  Eqn.~(\ref{eq:para_state}), adopted in the original QAOA, fully groups (FG) the trainable parameters and has the ansatz design $\mathcal{A}_{FG}=\{H_M, H_C\}$, which is symmetric with respect to $H_C$ under the \textit{permutation symmetry}. This is because its mixer Hamiltonian $H_M=\sum_{i=1}^N X_i$ is invariant under an arbitrary permutation operator. 

However, $\mathcal{A}_{FG}$ fails to employ the specific symmetry group of $H_C$. This issue can be addressed by partially grouping (PG) trainable parameters in QAOA. For example, denote $H_C=\sum_{(i_k,j_k)} Z_{i_k}Z_{j_k}$ with $Z_{j_k}$ being the Pauli-Z operator acting on the $j_k$-th qubit. An alternative symmetric ansatz design is $\mathcal{A}_{PG}=\{H_{\mathcal{O}_1}, \cdots, H_{\mathcal{O}_{|\mathcal{O}|}}, H_{\mathcal{O}_1^{e}}, \cdots, H_{\mathcal{O}_{|\mathcal{O}^e|}^{e}}\}$ where $H_{\mathcal{O}_k}=\sum_{i \in \mathcal{O}_k} X_i$ and $H_{\mathcal{O}_k^{e}}=\sum_{(i,j) \in \mathcal{O}_k^{e}} Z_iZ_j$ refer to the generators respecting the permutation symmetry of $H_C$ satisfying $H_M= \sum_{j=1}^{|\mathcal{O}|}H_{\mathcal{O}_j}$ and $H_C=\sum_{j=1}^{|\mathcal{O}^e|}H_{\mathcal{O}_j^e}$ \cite{sauvage2022building}. The ansatz design $\mathcal{A}_{PG}$ enables more free parameters than the ansatz design $\mathcal{A}_{FG}$ in each layer, and has been empirically shown with a faster convergence rate than $\mathcal{A}_{FG}$ given the same number of layers. 

When $H_C$ is asymmetric, another typical ansatz design in QAOA is $\mathcal{A}_{NG}=\{Z_{i_1}Z_{j_1}, \cdots, Z_{i_k}Z_{j_k}, X_1, \cdots, X_N\}$, where the parameters of all parameterized gates are independent and non-grouping (NG). Notably, the PQCs related to various ansatz design $\mathcal{A}_{FG},\mathcal{A}_{PG},\mathcal{A}_{NG}$ employ the same parameterized gates but with different \textit{parameter grouping strategies}, where $(\bm{\beta}, \bm{\alpha})$ in each layer can be fully grouped, partially grouped, and non-grouped \cite{sauvage2022building}.

\section{Convergence theory of QAOA}\label{sec:theory}
In this section, we theoretically illustrate how employing appropriate parameter grouping corresponds to better convergence performance. Similar to Refs.~\cite{you2022convergence} and~\cite{wang2022symmetric}, our derivations are based on the observation that the exploited PQC with highly-symmetric ansatz structure generally enables a faster convergence rate.

\begin{theorem}[Convergence]\label{thm:main_convergence}
    Consider a QAOA instance denoted as ($\ket{\psi_0},U(\bm{\theta}),H_C$) with $U(\bm{\theta})$ determined by the related ansatz design. Let $\mathcal{A}_{FG}, \mathcal{A}_{PG}, \mathcal{A}_{NG}$ be the ansatz designs of the circuits with parameters fully grouped, partially grouped, and no-grouped. Their effective dimension  yields
    \begin{equation}\label{eq:DLA_dim_relation_main}
        d_{\eff}(\mathcal{A}_{FG}) = d_{\eff}(\mathcal{A}_{PG}) \le d_{\eff}(\mathcal{A}_{NG}),
    \end{equation}
    where the equality in the inequality holds if there is no spatial symmetry in $H_C$. Besides, there exists a $d_{\eff}$-dependent threshold $C$ so that circuit depth $p>C$, the iterations $T$ required to achieve the same approximation ratio yield
    \begin{equation}\label{eq:Iter_steps_main}
        T_{PG} = T_{FG} \le T_{NG}.
    \end{equation}
\end{theorem}

The proof of Theorem~\ref{thm:main_convergence} and more elaborations are presented in Appendix~\ref{sec:proof}. The achieved results, combined with the over-parameterization theory of PQCs \cite{larocca2023theory}, deliver the following two implications. First, when the circuit depth $p>C$ is sufficiently large such that all PQCs with various ansatz designs reach the \textit{over-parameterization regime}, performing the parameter grouping can effectively decrease the effective dimension $d_{\eff}$ compared with the PQCs with no-parameter grouping, leading to a faster convergence rate. Second, the over-parameterization of QAOA occurs when the number of trainable parameters exceeds a critical point that is proportionally related to $d_{\eff}$. 

The above two implications indicate the selection of $\mathcal{A}_{FG}$, $\mathcal{A}_{PG}$, or $\mathcal{A}_{NG}$ is complicated and is both \textit{depth- and problem-dependent}. In particular, given a specified $p$, adopting a parameter grouping strategy can simultaneously reduce the number of parameters and the effective dimension, making it difficult to determine whether the QAOA reaches the over-parameterization regime. For instance, in a scenario such that the parameter grouping strategy drastically reduces the number of parameters but only slightly reduces the effective dimension, an over-parameterized QAOA could transform to an under-parameterized QAOA, leading to a degraded convergence as the optimization can be easily stuck in bad local minimal \cite{you2021exponentially,anschuetz2021critical}.

\section{MG-Net}\label{sec:mgnet}
The implication of  Theorem~\ref{thm:main_convergence}  inspires us to devise a method for dynamically generating an appropriate mixer Hamiltonian $H_M$ tailored to both the problem $G$ at hand and the specified circuit depth $p$. For this purpose, we harness the power of deep learning and devise an end-to-end learning framework, dubbed Mixer Generator Network (MG-Net).

\begin{figure*}[htp]
    \centering
    \includegraphics[width=1\linewidth]{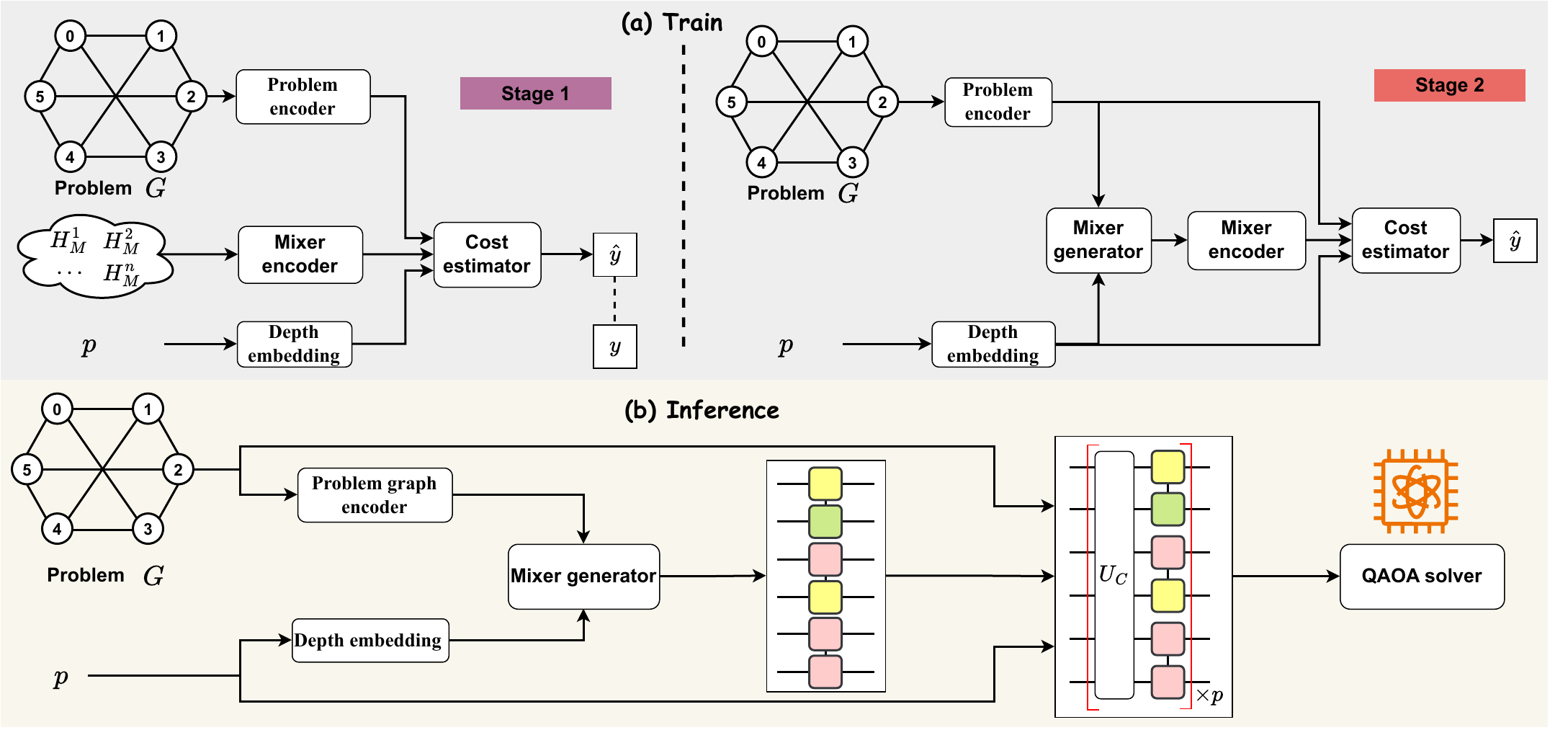}
    \caption{\small{\textbf{Framework of MG-Net.} (a) Training Phase. Initially (left), the cost estimator is trained to precisely predict QAOA performance for specific problem instances, circuit depths, and mixer Hamiltonians. In the subsequent stage (right), with the cost estimator fixed, the mixer generator is trained through unsupervised learning to derive the optimal mixer Hamiltonian that minimizes the cost estimator's output. (b) Inference Phase.  Given a problem $G$ and circuit depth $p$, the mixer generator produces a mixer Hamiltonian, subsequently utilized in a QAOA solver to find the solution.}}
    \label{fig:framework}
\end{figure*}

\subsection{Framework of MG-Net}\label{subsec:frame}
Before presenting the proposed MG-Net, let us first formalize the learning problem towards designing the mixer Hamiltonian $H_M$. To incorporate different Pauli operators and parameter grouping strategies, we extend the definition of an $N$-qubit mixer Hamiltonian $H_M$ in Eqn.~(\ref{eq:para_state}) to a more generalized form, supporting flexible operators and parameter correlations by substituting the Pauli-X operator with a selection of general Pauli operator and stratifying the $N$ operators into $K$ groups. Mathematically, the refined mixer Hamiltonian yields
\begin{equation}\label{eq:hm}
    H_M=\sum_{j=1}^K\beta_j\sum_{i\in \mathcal{G}_j}P_i,
\end{equation}
where $\beta_j$ refers to the trainable parameter controlling the $j$-th group of operators, $P_i\in \{X_i,Y_i\}$, and $\mathcal{G}_j$ contains the indices of operators belonging to the $j$-th group such that $\cup_{j=1}^K\mathcal{G}_j=[N]$ and $\mathcal{G}_i\cap\mathcal{G}_j=\emptyset$ for $\forall i\neq j$. In this sense, operators in the same group are correlated with each other, sharing the same parameter. In this way, \textit{the design of $H_M$ is decoupled into two distinct tasks: determine the parameter groups $\{\mathcal{G}_j\}_{j=1}^K$; identify the appropriate operator types $P_i$}. With the reformulation above, the decoupled tasks can be accomplished by learning a mapping rule $f: (G,p)\rightarrow (\mathcal{G},\mathcal{P})$ with $\mathcal{G}=\{\mathcal{G}_j\}_{j=1}^K$ and $\mathcal{P}\in \{X,Y\}^{\otimes N}$ referring to the parameter correlation and mixer Hamiltonian.

Designing a model to learn $f$ faces \textit{two main challenges}:\\ \textbf{(C-1)} The variety of combinatorial optimization tasks leads to uncertain input formats for the model, which  necessitates a universal representation method and retains essential properties of the original data, such as permutation invariance; \\
\textbf{(C-2)} The exponential growth of the search space for both parameter correlation and operator types, (i.e., scaling at $O(N^N)$ and $O(2^N)$, respectively), hurdles the design of an effective learning method. For instance, directing training a learning model in the supervised learning paradigm may require computationally unaffordable training examples to ensure good prediction accuracy.

We next present an end-to-end learning framework---\textbf{M}ixer \textbf{G}enerator \textbf{Net}work (MG-Net), as depicted in Fig.~\ref{fig:framework}, to address the above challenges. Particularly, to address  \textbf{C-1}, we devise a problem encoder which transforms each problem $G$ into a unified directed acyclic graph $G_C$, ensuring a consistent and effective input format. Coupled with the mixer encoder, it maps both the problem and mixer Hamiltonian to a shared hidden space. To address  \textbf{C-2}, MG-Net features a unique \textit{estimator-generator} framework, supplemented by a \textit{two-stage training strategy}. The role of these techniques is summarized below and their implementation details are demonstrated in the subsequent subsections. 
 
\noindent\textbf{Role of estimator}. Rather than directly seeking the optimal parameter correlation strategy $\mathcal{G}^*$ and operator type $ \mathcal{P}^*$ for a given $(G, p)$, we devise a \textit{cost estimator} to map the relationship between $(\mathcal{G}, \mathcal{P})$ and the achievable minimal cost $F_p$ of the corresponding QAOA in Eqn.~(\ref{eqn:QAOA-cost}). 

\noindent\textbf{Role of generator}. We devise a \textit{generator} to predict $(\mathcal{G}, \mathcal{P})$ that minimizes the cost estimator's output. This design requires only the cost of any mixer Hamiltonian as a label, thus avoiding the exhaustive search of optimal pairs $(\mathcal{G}^*,\mathcal{P}^*)$.

\textbf{Two-stage training}. The pipeline is visualized in Fig.~\ref{fig:framework}(a). \\ 
$\bullet$ \textbf{Stage 1 (Cost Estimator Training)}. This stage, marked in purple, focuses on training the cost estimator using \textit{supervised learning}. Inputs include the problem graph $G$, potential mixer Hamiltonians $H_M$, and the chosen circuit depth $p$, with the corresponding cost $y$ as the target label. \\ 
$\bullet$ \textbf{Stage 2 (Mixer Generator Training)}. This stage, marked in orange, freezes the cost estimator and only updates the mixer generator to minimize the output of the cost estimator under the \textit{unsupervised learning} paradigm.

For inference on unknown problem instances (in Fig.~\ref{fig:framework}(b)), MG-Net employs only the mixer generator to predict the optimal mixer Hamiltonian, which is then fed into a QAOA solver to derive the final solution. Distinguished by its ability to generalize effectively across a class of problems from a limited learning set, MG-Net sets itself apart from previous studies.
Refer to Appendix.~\ref{app:related_work} for discussion.
 
\subsection{Implementation of MG-Net}

\textbf{Data encoder in MG-Net.} MG-Net exploits three types of data encoder, i.e.,  the problem encoder, mixer encoder, and depth encoder, which maps the given problem $G$, the candidate mixer Hamiltonian $H_M$, and the specified depth $p$ to the same hidden feature space. The construction of these encoders is introduced below and the omitted details are deferred to Appendix \ref{app:implement_detail_de}.

\textbf{Cost estimator in MG-Net (Stage 1).} Recall Stage 1 in Sec.~\ref{subsec:frame}, the cost estimator takes the encoded problem graph $G_C$, the encoded mixer Hamiltonian $G_M$, and the encoded circuit depth $\bm{x}_p$ as inputs, and outputs the prediction of the achievable minimum loss of the corresponding QAOA. Each input is processed by an independent branch respectively: \textit{the problem graph branch, the mixer Hamiltonian branch, and the circuit depth branch}, as shown in Fig.~\ref{fig:net}(a). The concatenation of three types of features is subsequently utilized by a multi-layer perceptron (MLP) to output the minimum loss $\hat{y}$ that the QAOA ansatz can achieve. Refer to Appendix.~\ref{app:net_arch} for details. 

\textbf{Mixer generator in MG-Net (Stage 2).} The mixer generator in MG-Net takes $G_C$ and $\bm{x}_p$ as input and outputs a targeted mixer Hamiltonian $H_M$. Specifically, the mixer generation is composed of two separate sub-generators: the operator type generator and the parameter grouping generator defined in Eqn.~(\ref{eq:hm}), shown in Fig.~\ref{fig:net}(b). The operator type generator is responsible for generating operator types $\mathcal{P}$, which is conceptualized as a graph node classification task. The parameter grouping generator is responsible for predicting the sets of index groups $\{\mathcal{G}_j\}_{j=1}^K$ with an unspecified $K$, which is modeled as a link prediction task. Refer to Appendix.~\ref{app:net_arch} for details.

\begin{figure}[htp]
    \centering
    \includegraphics[width=0.9\linewidth]{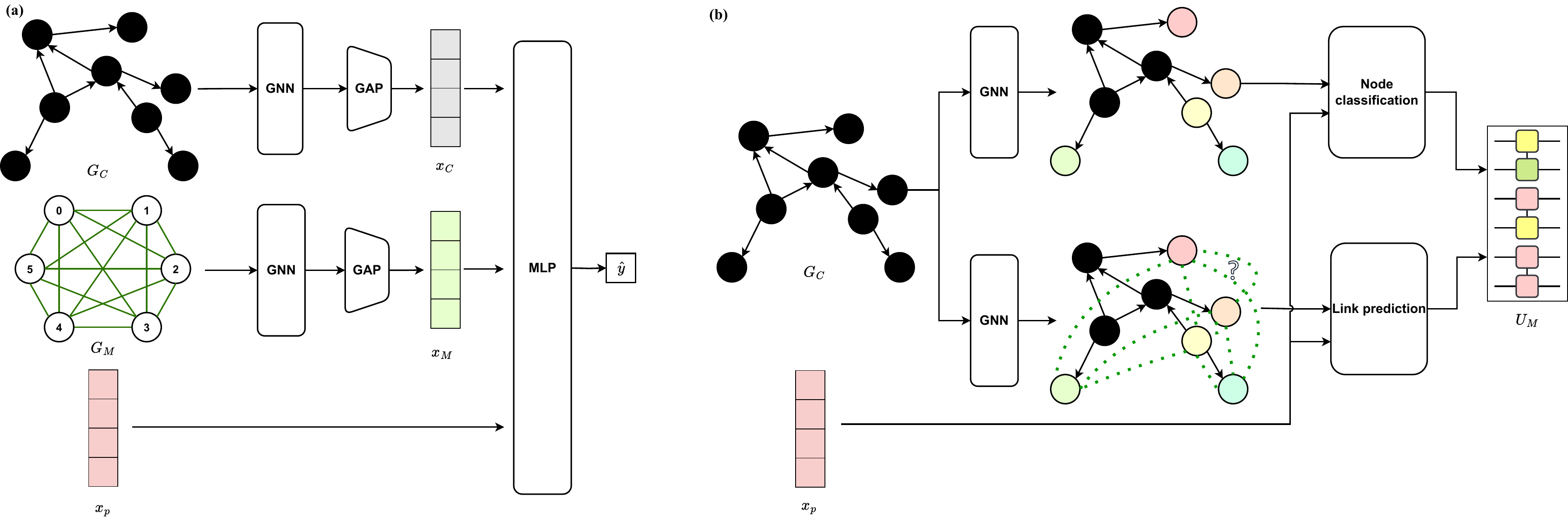}
    \caption{\small{\textbf{Structure of cost estimator and mixer generator.} (a) Cost estimator. The cost estimator is comprised of three distinct branches, each dedicated to processing different types of data: the original problem, the candidate mixer Hamiltonian, and the circuit depth. Their outputs are then integrated to predict the cost value achievable by the QAOA circuit. (b) Mixer generator. The mixer generation is divided into two distinct parts: operator type generation and parameter grouping generation. The former is executed as a node classification task, while the latter is approached as a link prediction task.}}
    \label{fig:net}
\end{figure}

\subsection{Training strategy}\label{subsec:train_strategy}

The training process of MG-Net is varied for the first and second stages, under supervised and unsupervised learning paradigms, respectively. 

\noindent\textbf{First-stage training}. This stage involves constructing a labeled dataset $\mathcal{D}_{\rm ce}^{\rm Tr}=\{(G_C^{(i)},G_M^{(i)},\bm{x}_p^{(i)}), y^{(i)}\}_{i=1}^S$, where the $i$-th sample consists of a tuple of features (i.e., the problem description $G_C^{(i)}$, the mixer $G_M^{(i)}$,  and the circuit depth feature $\bm{x}_p^{(i)}$), and the label $y^{(i)}$ representing the minimum cost value achievable by this QAOA instance (i.e.,  determined by repeatedly executing such a QAOA with varying initial parameters). Once $\mathcal{D}_{\rm ce}^{\rm Tr}$ is ready, the cost estimator is optimized by minimizing the loss function 
\begin{equation}\label{eq:loss-stage1}
    \mathcal{L}_{\rm ce}=\lambda_{e}\mathcal{L}_{e} + \lambda_{r}\mathcal{L}_{r},
\end{equation}
where  $\lambda_{e}\in [0,1]$ and $\lambda_{r}\in [0,1]$ are two hyper-parameters of each loss, $\mathcal{L}_{e}=\frac{1}{S}\sum_{i=1}^S(y^{(i)}-\hat{y}^{(i)})^2$ is the mean square error, and $\mathcal{L}_{r}$ is the ranking loss
\begin{equation}
    \mathcal{L}_{r}=\frac{1}{S^2-S}\sum_{i,j}^S\max (0,1-{\rm sign}(y^{(i)}-y^{(j)})(\hat{y}^{(i)}-\hat{y}^{(j)})). \nonumber
\end{equation}

\noindent\textbf{Second-stage training}. This stage involves the training of the mixer generator via unsupervised learning. The loss function of this stage is  
\begin{equation}
    \mathcal{L}_{\rm mg}=\frac{1}{S}\sum_{i=1}^SC(G_C^{(i)},M(G_C^{(i)},\bm{x}_p^{(i)}),\bm{x}_p^{(i)}),
\end{equation}
where $C(\cdot)$ and $M(\cdot)$ represent the output of the cost estimator and mixer generator, respectively. Note that only the parameters of the mixer generator are updated; the cost estimator parameters remain fixed to ensure consistent evaluation criteria throughout the whole learning process.

\section{Experiments}\label{sec:exp}
We evaluate the performance of MG-Net by two typical applications of QAOA: weighted Max-Cut and Transverse-field Ising model (TFIM), each of which is elucidated below.

\textbf{Weighted Max-Cut.} Denote a weighted graph  as $G=(V,E,W)$, where $V$ is the set of vertices of graph, $E$ is the set of graph edges, $W=\{w_{ij}\}_{(i,j)\in E}$ is the set of weights assigned to each edge. The problem Hamiltonian for the weighted Max-Cut problem is $H_C^{{\rm MaxCut}}=0.5*\sum_{(i,j)\in E}w_{ij}Z_iZ_j$,
where $Z_i$ is a Pauli-Z operator acting on the $i$-th qubit.

\textbf{TFIM.} Our focus is a class of inhomogeneous TFIMs: $H_C^{{\rm TFIM}}=-\sum_{(i,j)}J_{ij}Z_iZ_j-h\sum_i X_i$,
where $J_{ij}$ is the interaction strength between neighboring spins (or qubits) $(i,j)$, and $h$ signifies the strength of a global transverse field applied to each spin. In this model, the interaction strengths $J_{ij}$ can vary between different pairs of spins, adding a layer of complexity to the system.

\subsection{Experiment configuration}

\textbf{Dataset construction.}The Max-Cut problem focuses weighted 3-degree regular (w3r) graphs, where the edge weights $\{w_{ij}\}$ are uniformly sampled from $[0,1]$. The TFIM focuses on 1D instances where a qubit $i\in [N-1]$ has neighbors $i \pm 1 \pmod{N}$. The strength $J_{ij}$ and $h$ are uniformly sampled from $[0.5, 1.5]$ and $[0.1, 2]$ respectively. The training dataset $\mathcal{D}_{\rm ce}^{\rm Tr}$ in Sec.~\ref{subsec:train_strategy} contains $S=100$ instances for both two tasks with size up to $N=64$ qubits, while The test dataset $\mathcal{D}^{\rm Te}$ contains another $100$ problem instances which are different from that of $\mathcal{D}_{\rm ce}^{\rm Tr}$. Refer to Appendix~\ref{app:implement_detail_dc} for details.

\textbf{Optimization and training of MG-Net}. The cost estimator and mixer generator are trained using an Adam optimizer with a learning rate of $10^{-4}$, and hyper-parameters $\lambda_{e}=1$ and $\lambda_{r}=1$ in Eqn.~(\ref{eq:loss-stage1}). 

\textbf{Optimization of QAOA.} After predicting the problem-hardware-tailored mixer Hamiltonian $H_M$ by the trained mixer generator, a QAOA circuit with the initial state $\ket{+}^{\otimes N}$ and $H_M$ is optimized by an Adam optimizer with a learning rate of $0.15$. Each setting undergoes $10$ independent runs with varied random seeds and initial parameters to obtain the statistical results. Refer to Appendix~\ref{app:implement_detail_dc} for details.

\subsection{Results}

\textbf{Cost estimator acts as an accurate performance indication for QAOA}. The behavior of the cost estimator on the test dataset with varying circuit depths $p$ and two distinct parameter grouping strategies NG and FG (defined in Theorem \ref{thm:main_convergence}) is recorded in Fig.~\ref{fig:loss_pred_performance}. In Fig.~\ref{fig:loss_pred_performance}(a), we observed a strong correlation between the estimated and minimum cost values, and the correlation strength changes with $p$ and parameter grouping strategy. Particularly, the cost estimator predicts a high likelihood of finding the most accurate solution for QAOA circuits with FG parameters and a depth of $p=92$. This prediction aligns with the actual performance of QAOA under these specific conditions.

\begin{figure}[htp]
    \centering
    \includegraphics[width=1.0\linewidth]{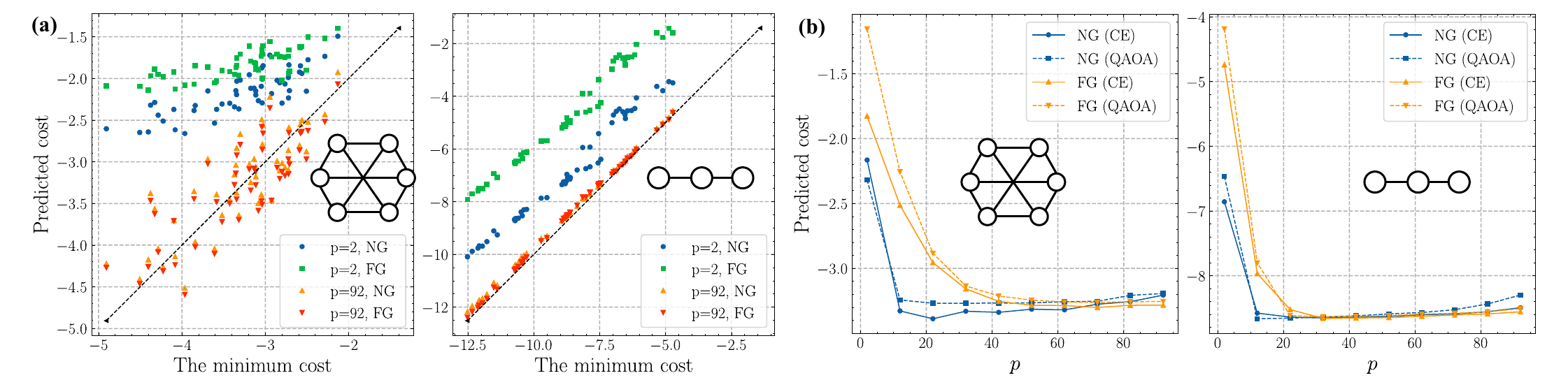}
    \caption{\small{\textbf{Behavior of cost estimator}. (a) The correlation between the estimated cost and the minimum cost for Max-Cut (left) and TFIM (right). Each point represents the result of a problem instance. The dashed line represents that QAOA can find the exact solution $y=x$. (b) The achievable cost under various circuit depth $p$ for Max-Cut (left) and TFIM (right). The label `CE' is the abbreviation of cost estimator. The dashed lines represent the cost achieved by QAOA, while the solid lines represent the cost estimated by our model.}}
    \label{fig:loss_pred_performance}
\end{figure}

We next focus on the behavior of the cost estimator concerning $p$ as shown in Fig.~\ref{fig:loss_pred_performance}(b). We note that for FG (standard QAOA), the estimated loss decreased monotonically with increasing $p$, aligning with standard QAOA's behavior. Under the NG scenario (multi-angle QAOA), a transition that QAOA performance begins to decline is observed when the circuit becomes excessively long ($p>42$). These results indicate the reliability of the cost estimator as a performance indicator for QAOA and reveal the complexities in QAOA performance under conditions of increased circuit length.

\textbf{Mixer generator}. We next evaluate the performance of the customized mixer Hamiltonian generated by MG-Net. As shown in Fig.~\ref{fig:mixer_gen_performance}(a), the number of trainable parameters $\#P$ of the generated quantum circuits aligns with the maximum in scenarios where all parameters are non-correlated (labeled as `NG') for smaller circuit depths $p<20$. This alignment indicates that MG-Net effectively enhances the expressibility of the QAOA ansatz for limited-depth circuits without significantly increasing the number of parameters, thereby avoiding potential trainability issues. As $p$ increases, a transition occurs. The growth rate of $\#P$ starts to decelerate, reaching a notable transition point at $p=62$ for Max-Cut ($p=52$ for TFIM). Beyond this threshold, the generated mixer Hamiltonians gradually converge towards the configuration seen in standard QAOA, with fully grouped parameters.

\begin{figure}[htp]
    \centering
    \includegraphics[width=1.0\linewidth]{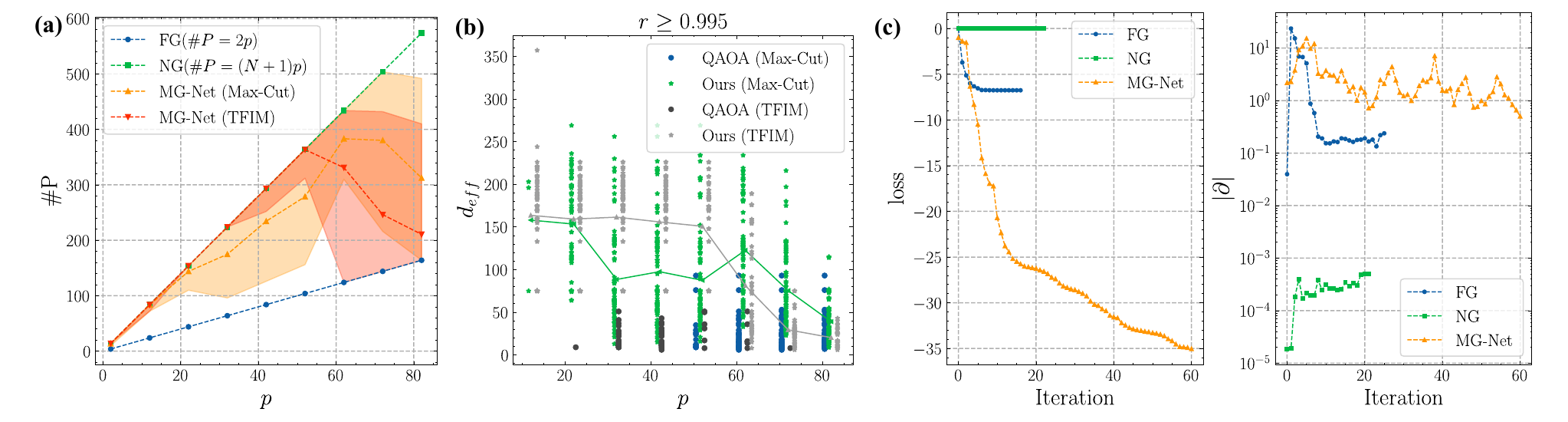}
    \caption{\small{\textbf{The trainability of the quantum circuits generated by MG-Net for Max-Cut and TFIM.} (a) The number $\#P$ of trainable parameters of the quantum circuits with mixer Hamiltonian predicted by MG-Net. (b) Comparison of the effective dimension $d_{\eff}$ of quantum circuits in standard QAOA and MG-Net driven QAOA (labeled as `Ours'). The green and grey solid lines denote the average effective dimension $d_{eff}$ of the predicted circuits that can achieve an approximation ratio over $0.995$ for Max-Cut and TFIM, respectively. It assesses circuits achieving an approximation ratio $r$ of at least $0.995$. (c) The convergence of QAOA with FG, NG and mixer Hamiltonian predicted by MG-Net for Max-Cut on $64$-node weighted graphs.}}
        \label{fig:mixer_gen_performance}
\end{figure}

Fig.~\ref{fig:mixer_gen_performance}(b) compares the effective dimension $d_{\eff}$ of quantum circuits achieving high approximation ratio $r\geq 0.995$ in standard QAOA and MG-Net driven QAOA. The results show that circuits generated by MG-Net achieve $r\geq 0.995$ across all values of $p$, even as low as  $p=2$, outperforming standard QAOA, which only reaches this level for $p>50$ for Max-Cut ($p>20$ for TFIM). Besides, the effective dimension of these high-quality quantum circuits gradually decreases with growing $p$, in line with the convergence analysis in Theorem~\ref{thm:main_convergence}. These findings suggest that MG-Net dynamically adjusts quantum circuits in response to changes in circuit depth $p$, thereby consistently ensuring high performance. 

Fig.~\ref{fig:mixer_gen_performance}(c) explicitly demonstrates the optimization behavior of 64-qubit QAOA with FG, NG and the mixer Hamiltonian predicted by our MG-Net. The left panel displays the loss curves during the optimization of quantum circuits with $p=2$, revealing that our method achieves the most rapid convergence. The right panel further explores the gradients of the three methods during optimization. Notably, the parameter gradient norm of our method maintains a trainable level of $1$, whereas the gradient for FG and NG falls to $10^{-1}$ and $10^{-4}$, respectively, compromising their trainability.

\textbf{Performance comparison}. In evaluating the effectiveness of our proposed method for solving Max-Cut problems, we conducted a comparative analysis against both classical and quantum algorithms. The benchmarks included the greedy algorithm, the Goemans-Williamson (GW) algorithm \cite{goemans1995improved}, alongside various quantum approaches such as QAOA, ADAPT-QAOA, and multi-angle QAOA (ma-QAOA). Our analysis, based on the average results from $100$ graphs in our test dataset, is summarized in Tab.~\ref{tab:ar}. The findings reveal that our method consistently outperforms other techniques in achieving a higher approximation ratio, particularly in larger-scale problems. Refer to Appendix~\ref{app:comp_tfim} for comparison results on TFIM. 

\textbf{More numerical results.} We have conducted additional analysis on the behavior of MG-Net and additional experiments on more tasks. Refer to \ref{app:more_result} for more details.

\begin{table}[htp]
    \centering
    \caption{\small{\textbf{Comparison of approximation ratio $r$ among different methods for Max-Cut.}}}
    \begin{tabular}{lccc}
        \toprule[1pt]
        Method & $6$ qubits & $16$ qubits & $64$ qubits  \\
        \midrule
        Greedy & $0.89\pm 0.104$ & $0.91\pm 0.047$ & 0.79 \\
        GW & $0.94\pm 0.074$ & $0.93\pm 0.052$ & 0.91\\
        \midrule
        QAOA & $0.93\pm 0.027$ & $0.35\pm 0.119$ & 0.19 \\
        ADAPT-QAOA & $0.75\pm 0.129$ & $0.58\pm 0.154$ & - \\
        ma-QAOA & $0.98\pm 0.004$ & $0.84\pm 0.129$ & 0.0 \\
        \midrule
        \textbf{Ours} & $\bm{0.99\pm 0.0004}$ & $\bm{0.95\pm 0.152}$ & \textbf{0.96} \\
        \bottomrule[1pt]
    \end{tabular}
    \label{tab:ar}
\end{table}

\section{Conclusion}\label{sec:discussion}
 
In this study, we analyze QAOA's convergence on varied mixer Hamiltonians, focusing on parameter grouping strategies. We introduce MG-Net for dynamically generating optimal mixer Hamiltonians for various problems and circuit depths. Numerical experiments on Max-Cut and TFIM confirm MG-Net's efficacy in enhancing QAOA's approximation ratio, particularly for large-scale problems, while ensuring low circuit complexity. This research advances the understanding and application of QAOA across various circuit depths.

Despite these promising outcomes, our work has several limitations that need to be addressed in future research. Firstly, training the cost estimator of MG-Net involves the construction of a labeled dataset $\mathcal{D}_{\rm ce}^{\rm Tr}$, which introduces additional resource consumption. Future work can focus on more efficient training algorithms. Additionally, our current approach is specifically designed for QAOA on early fault-tolerant devices, which limits the exploration of extending MG-Net to other quantum algorithms and noisy devices. Addressing these limitations will further enhance the robustness and scalability of MG-Net,  offering potential for broader use in VQAs.

% {
%     \small
%     \bibliographystyle{plainnat} 
%     \bibliography{ref.bib} 
% }

\bibliographystyle{unsrtnat} 
\bibliography{ref.bib}

%%%%%%%%%%%%%%%%%%%%%%%%%%%%%%%%%%%%%%%%%%%%%%%%%%%%%%%%%%%%

\appendix

\section{Optimization of QAOA}\label{sec:appdix_Opt_QNNs}
	In this section, we separately elaborate on the elementary notations in quantum computing, the preliminary of Hamiltonian, and the optimization strategy of QAOA.
	
	\textbf{Basics of quantum computation.} The elementary unit of quantum computation is qubit (or quantum bit), which is the quantum mechanical analog of a classical bit. A qubit is a two-level quantum-mechanical system described by a unit vector in the Hilbert space $\mathbb{C}^2$. In Dirac notation, a qubit state is defined as $\ket{\phi}=c_0\ket{0}+c_1\ket{1}\in \mathbb{C}^2$ where $\ket{0}=[1,0]^{\top}$ and $\ket{1}=[0,1]^T$ specify two unit bases and the coefficients $c_0,c_1\in\mathbb{C}$ yield $|c_0|^2+|c_1|^2=1$. Similarly, the \textit{quantum state} of $n$ qubits is defined as a unit vector in $\mathbb{C}^{2^n}$, i.e., $\ket{\psi}=\sum_{j=1}^{2^n}c_j\ket{e_j}$, where $\ket{e_j}\in \mathbb{R}^{2^n}$ is the computational basis whose $j$-th entry is $1$ and other entries are $0$, and $\sum_{j=1}^{2^n}|c_j|^2=1$ with $c_j \in \mathbb{C}$. Besides Dirac notation, the density matrix can be used to describe more general qubit states. For example, the density matrix of the state $\ket{\psi}$ is $\rho=\ket{\psi}\bra{\psi} \in \mathbb{C}^{2^n \times 2^n}$, where $\bra{\psi}=\ket{\psi}^{\dagger}$ refers to the complex conjugate transpose of $\ket{\psi}$. For a set of qubit states $\{p_j, \ket{\psi_j}\}_{j=1}^m$ with $p_j>0$, $\sum_{j=1}^m p_j=1$, and $\ket{\psi_j}\in \mathbb{C}^{2^n}$ for $j \in [m]$, its density matrix is $\rho=\sum_{j=1}^m p_j \rho_j$ with $\rho_j=\ket{\psi_j}\bra{\psi_j}$ and $\Tr(\rho)=1$.
	
	A \textit{quantum gate} is a unitary operator that can evolve a quantum state $\rho$ to another quantum state $\rho^{\prime}$. Namely, an $n$-qubit gate $U\in\mathcal{U}({2^n})$ obeys $UU^{\dagger}=U^{\dagger}U=I_{2^n}$, where $\mathcal{U}({2^n})$ refers to the unitary group in  dimension $2^n$. Typical single-qubit quantum gates include the Pauli gates, which can be written as Pauli matrices:
		\begin{equation}
			X = \left[ \begin{array}{ccc}
				0 & 1 \\
				1 & 0 \\
			\end{array}
			\right], \quad 
			Y = \left[ \begin{array}{ccc}
				0 & -i \\
				i & 0 \\
			\end{array}
			\right], \quad 
			Z = \left[ \begin{array}{ccc}
				1 & 0 \\
				0 & -1 \\
			\end{array}
			\right]. \quad  \label{eq:pauli}
		\end{equation}
		The more general quantum gates are their corresponding rotation gates $R_X(\theta)=e^{-i\frac{\theta}{2}X}, R_Y(\theta)=e^{-i\frac{\theta}{2}Y}$, and $R_Z(\theta)=e^{-i\frac{\theta}{2}Z}$ with a tunable parameter $\theta$, which can be written in the matrix form as
		\begin{equation}
			R_X(\theta)=\left[\begin{array}{cc}
				\cos \frac{\theta}{2} & -i \sin \frac{\theta}{2} \\
				-i \sin \frac{\theta}{2} & \cos \frac{\theta}{2}
			\end{array}\right], 
			R_Y(\theta)=\left[\begin{array}{cc}
				\cos \frac{\theta}{2} & -\sin \frac{\theta}{2} \\
				\sin \frac{\theta}{2} & \cos \frac{\theta}{2}
			\end{array}\right],  
			R_Z(\theta)=\left[\begin{array}{cc}
				e^{-i \frac{\theta}{2}} & 0 \\
				0 & e^{i \frac{\theta}{2}}
			\end{array}\right]. \label{eq:pauli_rot}
		\end{equation}
		They are equivalent to rotating a tunable angle $\theta$ around $x$, $y$, and $z$ axes of the Bloch sphere, and recovering the Pauli gates $X$, $Y$, and $Z$ when $\theta=\pi$. Moreover, a multi-qubit gate can be either an individual gate (e.g., CNOT gate) or a tensor product of multiple single-qubit gates.  
	
	The \textit{quantum measurement} refers to the procedure of extracting classical information from the quantum state. It is mathematically specified by a Hermitian matrix $H$ called the \textit{observable}. Applying the observable $H$ to the quantum state $\ket{\psi}$ yields a random variable whose expectation value is $\bra{\psi}H\ket{\psi}$. 
	
	\textbf{Hamiltonian and ground state}. 
	In quantum computation, a \textit{Hamiltonian} is a Hermitian matrix that is used to characterize the evolution of a quantum system or as an observable to extract the classical information from the quantum system. Specifically, under the Schr\"odinger equation, a quantum gate has the mathematical form of $U=e^{-itH}$, where $H$ is a Hermitian matrix, called the Hamiltonian of the quantum system, and $t$ refers to the evolution time of the Hamiltonian. Typical single-qubit Hamiltonians include the Pauli matrices defined in Eqn.~(\ref{eq:pauli}). As a result,  the evolution time $t$ refers to the tunable parameter $\theta$ in Eqn.~(\ref{eq:pauli_rot}). Any single-qubit Hamiltonian can be decomposed as the linear combination of Pauli matrices, i.e., $H=a_1I+a_2X+a_3Y+a_4Z$ with $a_j \in \mathbb{C}$. In the same way, a multi-qubit Hamiltonian is denoted by $H=\sum_{j=1}^{4^n}a_jP_j$, where $P_j\in\{I,X,Y,Z\}^{\otimes n}$ is the tensor product of Pauli matrices. In quantum chemistry and quantum many-body physics, the Hermitian matrix that describes the quantum system to be solved is denoted as the \textit{problem Hamiltonian} $H_C$. Within the context of QAOA, the information of the graph is encoded in the problem Hamiltonian, which is also called cost Hamiltonian. Another essential Hamiltonian in QAOA refers to the mixer Hamiltonian $H_M$, which is designed to facilitate transitions between different states (solutions), allowing the algorithm to explore the solution space. 
	
	When taking the problem Hamiltonian as the observable, the quantum state $\ket{\psi^*}$ is said to be the \textit{ground state} of problem Hamiltonian $H$ if the expectation value $\bra{\psi^*}H\ket{\psi^*}$ takes the minimum eigenvalue of $H$, which is called the \textit{ground energy}. The solution of the optimization problem is encoded in the ground state of the problem Hamiltonian.

	\textbf{Optimization of QAOA.} The loss function for QAOA with problem Hamiltonian $H_C$ is generally defined as 
    \begin{equation}\label{eq:qaoa_loss}
        \mathcal{L}(\bm{\theta}=(\bm{\alpha}, \bm{\beta}))=\braket{\psi_0|U(\bm{\theta})^\dagger H_C U(\bm{\theta})|\psi_0},
    \end{equation}
    where $U(\bm{\theta})$ refers to the parameterized unitary implemented on a quantum computer and $\ket{\psi_0}$ is an easily prepared state, which is generally set as the computational basis state $\ket{0^{\otimes n}}$.
    The optimization of the loss function  $\mathcal{L}(\bm{\theta})$ can be completed by gradient-based methods. A plethora of optimizers have been designed to estimate the optimal parameters $\bm{\theta}^*=\min_{\bm{\theta}} \mathcal{L}(\bm{\theta})$. Here we introduce the implementation of the first-order gradient-based optimizer for self-consistency. Refer to \cite{cerezo2021variational} for a comprehensive review. 
	
	Based on Eqn.~(\ref{eq:para_state}), the trainable parameters of QAOA are denoted by $\bm{\theta}=(\bm{\theta}_1^{\top}, \cdots, \bm{\theta}_L^{\top})^{\top}$ with $\bm{\theta}_{\ell}=(\theta_{\ell 1}, \cdots, \theta_{\ell K})^T$, where the subscript `$\ell k$' refers to the $k$-th parameter of the $\ell$-th layer $U_{\ell}$ for $\forall k\in[K]$ and $\forall \ell \in [L]$. The corresponding update rule at the $t$-th iteration $\forall t\in[T]$ is 
	\begin{eqnarray}
		&&\bm{\theta}^{(t+1)} \nonumber\\
		= && \bm{\theta}^{(t)}-\eta \frac{\partial \mathcal{L}(\bm{\theta}^{(t)})}{\partial \bm{\theta}} \nonumber\\
		= && \bm{\theta}^{(t)}-\eta \left(\bra{\psi_0}U(\bm{\theta}^{(t)})^{\dagger} H_C U(\bm{\theta}^{(t)})\ket{\psi_0}  - E_0 \right)\frac{\partial \left(\bra{\psi_0}U(\bm{\theta}^{(t)})^{\dagger} H_C U(\bm{\theta}^{(t)})\ket{\psi_0}  - E_0 \right)}{\partial \bm{\theta}}, \nonumber
	\end{eqnarray}
	where $\eta$ refers to the learning rate. The derivative in the last equality can be calculated via the parameter shift rule \cite{mitarai2018quantum}. Mathematically, the derivative with respect to the parameter ${\theta}_{\ell k}$ for $\forall \ell\in[L]$ and $\forall k\in[K]$ is 
	\begin{eqnarray}\label{eqn:para_shift}
		&& \frac{\partial \left(\bra{\psi_0}U(\bm{\theta})^{\dagger} H_C U(\bm{\theta})\ket{\psi_0}  - E_0 \right) } {\partial {\theta}_{\ell k}} \nonumber\\
		= && \frac{1}{2\sin \alpha} \big[\left(\bra{\psi_0}U(\bm{\theta}^+)^{\dagger} H_C U((\bm{\theta}^+)\ket{\psi_0}  - E_0 \right)  
		- \left(\bra{\psi_0}U((\bm{\theta}^-)^{\dagger} H_C U((\bm{\theta}^-)\ket{\psi_0}  - E_0 \right)\big]\nonumber,
	\end{eqnarray}
	where $ \bm{\theta}^+ = \bm{\theta} +  \alpha \bm{e}_{\ell k}$, $ \bm{\theta}^- = \bm{\theta} -  \alpha \bm{e}_{\ell k}$, $\bm{e}_{\ell k}$ is the unit vector along the $\theta_{\ell k}$ axis and $\alpha$ can be any real number but the multiple of $\pi$ because of the diverging denominator.

\section{Proof}\label{sec:proof}
The theoretical analysis of the convergence for symmetric QAOA is based on representation theory. In this regard, we first introduce the foundation of representation theory related to QAOA in Appendix~\ref{subsec:repre_qaoa}. The proof of  Theorem~\ref{thm:main_convergence} is elaborated in Appendix~\ref{subsec:proof_theorem}. 

\subsection{Representation theory in QAOA}\label{subsec:repre_qaoa}
In general, an instance of QAOA is specified by a triplet $(\ket{\psi_0}, U(\bm{\theta}), H)$, where $\ket{\psi_0}$ and $H$ refer to the initial state and problem Hamlitonian, and $U(\bm{\theta})$ refers to the parameterized quantum circuit (ansatz) with the form of
    \begin{equation}
        U(\bm{\theta}) = \prod_{j=1}^P\prod_{k=1}^K e^{-i\theta_{j,k}H_k},
\end{equation}
where $\bm{\theta}=(\bm{\theta}_{11},\cdots,\bm{\theta}_{1K}, \cdots, \bm{\theta}_{P1},\cdots,\bm{\theta}_{PK})\in\Theta \subseteq \mathbb{R}^{PK}$ is trainable parameters, $j$ is the index of layer, and $\mathcal{A}=\{H_k\}_{k=1}^{K}$ is set of Hermitian traceless operators called an ansatz design. The difference of ansatz originates from the varied $\Theta$ and $\mathcal{A}$. Given $\Theta$ and $\mathcal{A}$, a set of ansatz forms a subgroup of $SU(2^n)$ with $\mathcal{U}_{\mathcal{A}}=\cup_{L=0}^{\infty}\{U(\bm{\theta}):\bm{\theta}\in \Theta\}$, which can be characterized by dynamical Lie group with dynamical Lie algebra \cite{larocca2022diagnosing}

\begin{definition}[Dynamical Lie algebra and dynamical Lie group, \cite{larocca2022diagnosing}]
    Given an ansatz design $\mathcal{A}=\{H_1,\cdots,H_K\}$, the dynamical Lie algebra (DLA) $\mathfrak{g}$ is generated by the repeated nested commutators of elements in $\mathcal{A}$, i.e.,
    \begin{equation}
        \mathfrak{g}={\rm span}\braket{iH_1,...,iH_{K}}_{Lie},
    \end{equation}
    where $\braket{S}_{Lie}$ denotes the $Lie$ closure, i.e., the set obtained by repeatedly taking the nested commutators of the elements in $S$. The set of unitaries $\mathcal{U}_\mathcal{A}$ that can be generated by the ansatz design $\mathcal{A}$ is determined by its DLA through
    \begin{equation}
         \mathcal{U}_\mathcal{A}=e^{\mathfrak{g}}:=\{e^H, H\in \mathfrak{g}\}.
    \end{equation}
\end{definition}

Furthermore, the algebra structures of the ansatz design $\mathcal{A}$ can be characterized through the representation and the subrepresentation of Lie algebra on specific vector space.

\begin{definition}[Representation of Lie algebra] \label{def:rep_LA}
    Let $\mathfrak{g}$ be a Lie algebra on a finite-dimensional vector space $V$. A representation $r$ of $\mathfrak{g}$ acting on $V$ is a Lie algebra homomorphism $r:\mathfrak{g}\to\mathfrak{g}\mathfrak{l}(V)$, i.e., a linear map satisfying
    \begin{equation}
        r([X,Y])=[r(X),r(Y)], \quad  \mbox{for all~} X,Y\in \mathfrak{g}.
    \end{equation}
    The dimension of the representation $r$ is defined by $\dim(r) = \dim(V)$. If there exists a direct sum decomposition of $V$ into subspaces $V=V_1\oplus V_2 \oplus \cdots \oplus V_k$ such that $r(g)v_j\in V_j$ for any $v_j\in V_j$ and any $g\in \mathfrak{g}$, then $r_j:=r|_{V_j}$ is called the subrepresentation of $r$ on the vector space $V_j$. Moreover, $r_j$ is irreducible if there is no non-trivial invariant subspace of $V_j$. Then the representation of $\mathfrak{g}$ on the vector space $V=V_1\oplus V_2 \oplus \cdots \oplus V_k$ can be written as 
    \begin{equation}\label{eq:irrep}
        r(g)(v)=(r_1\oplus \cdots \oplus r_k(g))(v_1, \cdots, v_k)=(r_1(g)v_1, \cdots, r_k(g)v_k), \quad \mbox{for all~} g\in \mathfrak{g}, ~v\in V.
    \end{equation}
    The dimension of the representation with irreducible representation in Eqn.~(\ref{eq:irrep}) is $\dim(r)=\sum_{j=1}^k \dim(V_j)$
\end{definition}

The representation of DLA $\mathfrak{g}$ refers to the natural representation $r:\mathfrak{g}\to \mathfrak{g}$. In this regard, the dimension of DLA refers to $\dim(\mathfrak{g})=\dim(r)$. While the dimension of DLA is employed to characterize the threshold of over-parameterization \cite{larocca2023theory} and the barren plateau \cite{larocca2022diagnosing}, it does not take into account the symmetry structure of the ansatz and the initial state concerning the problem Hamiltonian. In particular, the symmetry operators of the DLA $\mathfrak{g}$ refer to unitary operators $S$ satisfying $SgS^{\dagger}=g$ for any $g\in \mathfrak{g}$, which is a subset of the commutant of $\mathfrak{g}$.

\begin{definition}[Commutant]
    Let $\mathfrak{g}$ be a matrix algebra. Its commutant is defined as $\mathcal{C}(\mathfrak{g}):=\{A: [A,g]=0, \forall g\in \mathfrak{g}\}$.
\end{definition}

We recall that the ansatz being symmetric with respect to the problem Hamiltonian means that there exists a symmetry group of the problem Hamiltonian $\mathcal{S}=\{S:S^{\dagger}H_CS=H_C\}$ such that $\mathcal{S}$ is also the symmetry group of the related DLA $\mathfrak{g}$, i.e., $\mathcal{S} \subseteq \mathcal{C}(\mathfrak{g})$. This indicates that the problem Hamiltonian and the ansatz design have the same block diagonalization structure \cite{schatzki2022theoretical}, namely the acting vector space $V=\oplus_{j=1}^k V_j$. Moreover, when there exists a subspace $V^* \in \{V_j\}_{j=1}^k$ such that the initial state lives in this space, then the optimization of the variational quantum state could be constrained into this subspace $V^*$ whose dimension refers to the effective dimension defined in Definition~\ref{def:eff_dim}. In this regard, the trainability of QAOA could be instead characterized by the effective dimension $d_{\eff}=\dim(V^*)$ \cite{wang2022symmetric,you2022convergence}. The relation between the effective dimension and the dimension of DLA is encapsulated in the following lemma.

\begin{lemma}[The relation between effective dimension and the dimension of DLA]\label{lem:eff_dla}
    Consider a QAOA instance ($\ket{\psi_0},U(\bm{\theta}),H_P$) with DLA $\mathfrak{g}$. If there exists an invariant subspace $V_{\mathfrak{g}}$ covering the initial state $\ket{\psi_0}$ and the solution state $\ket{\psi^*}=U(\bm{\theta}^*)\ket{\psi_0}$, then the effective dimension $d_{\eff}$ of this ansatz design $\mathcal{A}$ and the dimension of the corresponding DLA $\mathfrak{g}$ yields $d_{\eff}\le \dim(\mathfrak{g})$.
\end{lemma}
\begin{proof}[Proof of Lemma~\ref{lem:eff_dla}]
    The derivation of $d_{\eff}\le \dim(\mathfrak{g})$ could be directly obtained from the observation of $d_{\eff}\le\max_{j\in[k]}\dim(V_j)\le \sum_{j=1}^k V_j= \dim(\mathfrak{g})$.
\end{proof}

\subsection{Proof of Theorem~\ref{thm:main_convergence}}\label{subsec:proof_theorem}
The proof of Theorem~\ref{thm:main_convergence} employs the following lemmas, whose proofs are deferred to Appendix~\ref{subsec:proof_DLA_dim}.
\begin{lemma}[Convergence, adapted from Corollary 5.4 in \cite{you2022convergence}]\label{lem:convergence}
    Consider a QAOA instance denoted as ($\ket{\psi_0},U(\bm{\theta}),H_C$) with the effective dimension $d_{\eff}$. The unitary operator $U(\bm{\theta})$ follows the Haar distribution over special unitary matrices. Let $\ket{\psi^*}$ denote the solution state for problem Hamiltonin $H_C$ and $\ket{\psi^{(t)}}$ be the state at the $t$-th iteration. There exists an $d_{\eff}$-dependent over-parameter threshold $C(d_{\eff})$ and a $PK$-dependent learning rate $\eta(PK)$ so that if the number of the ansatz parameters $PK\geq C$, then with high probability, under gradient flow with learning rate $\eta(PK)$, the output state $\ket{\psi^{(t)}}$ converges to the solution state with error $\epsilon=1-|\braket{\psi^{(t)}|\psi^*}|$ after $T_{\epsilon}=O(\log \frac{d_{\eff}}{\epsilon})$ iterations.
\end{lemma}

\begin{lemma}\label{lem:DLA_dim_relation}
    Let $\mathcal{A}_{FG}, \mathcal{A}_{PG}, \mathcal{A}_{NG}$ be the ansatz designs of the circuits with parameters fully grouping, partially grouping, no-grouping, then the effective dimension related to $\mathcal{A}_{FG}, \mathcal{A}_{PG}, \mathcal{A}_{NG}$ yields
    \begin{equation}\label{eq:DLA_dim_relation}
        d_{\eff}(\mathcal{A}_{FG}) = d_{\eff}(\mathcal{A}_{PG}) \le d_{\eff}(\mathcal{A}_{NG}),
    \end{equation}
    where the equality in the inequality holds if there is no permutation symmetry in the problem Hamiltonian.
\end{lemma}

\begin{proof}[Proof of Theorem~\ref{thm:main_convergence}]
To obtain the ordering relation of the convergence rate of various ansatz designs, we first elucidate the relation between the convergence rate of the approximation ratio and the effective dimension. Consider the problem Hamiltonian $H_C=\sum_{(i,j)}Z_iZ_j \in \mathbb{C}^{d\times d}$ with $d=2^N$ and eigenvalues $\lambda_1\leq \lambda_2\cdots \leq \lambda_d$ and its corresponding eigenvector $\{\ket{\lambda_i}\}_{i=1}^d$. Preparing a quantum state $\ket{\psi}$ with overlap with the target ground state $\ket{\psi^*}$: $|\braket{\psi|\psi^*}|=1-\epsilon$, the lower bound of the expectation value of $\braket{\psi|H_C|\psi}$ is
\begin{align}
    \braket{\psi|H_C|\psi} &= \braket{\psi|\sum_{i=1}^d\lambda_i\ket{\lambda_i}\bra{\lambda_i}\psi}\\
    &=\lambda_1(1-\epsilon)^2+\sum_{i=2}^d\lambda_i|\braket{\psi|\lambda_i}|^2 \\
    &\leq \lambda_1(1-\epsilon)^2 + \lambda_d (1-(1-\epsilon)^2),
\end{align}
where the first inequality works by scaling each eigenvalue $\lambda_i$ to $\lambda_d$ and following the fact $\sum_{i=2}^d|\braket{\psi|\lambda_i}|^2 \le 1-(1-\epsilon)^2$. Then approximation ratio $r$ is
\begin{align}
    & r=\frac{\braket{\psi|H_C|\psi}}{\lambda_1}\geq \frac{\lambda_d}{\lambda_1}-\frac{\lambda_d-\lambda_1}{\lambda_1}(1-\epsilon)^2\ge (1-\varepsilon)^2, \nonumber \\
    \Longrightarrow \quad & \epsilon \le 1-\sqrt{r}
\end{align}
where the first inequality in the first equation holds because $\lambda_1 < 0$. Employing Lemma~\ref{lem:convergence}, we have that the output state $\ket{\psi^{(t)}}$ converges to the solution state with approximation ratio $r\ge |\braket{\psi^{(t)}|\psi^*}|^2$ after $T_r=O(\log(\frac{d_{\eff}}{1-\sqrt{r}}))$ iteration steps. These achieved results indicate that a small effective dimension leads to a faster convergence rate. In this regard, combining with Lemma~\ref{lem:DLA_dim_relation}, the convergence rate $T$ related to various ansatz $\mathcal{A}_{NG},\mathcal{A}_{PG},\mathcal{A}_{FG}$ for achieving the same approximation ratio yields $T_{FG}= T_{PG} \le T_{NG}$.
\end{proof}

\subsection{Proof of Lemma~\ref{lem:DLA_dim_relation}}\label{subsec:proof_DLA_dim} 
The proof of Lemma~\ref{lem:DLA_dim_relation} employs the following lemmas, where the proofs of Lemma~\ref{lem:rep_sum_subrep} and Lemma~\ref{lem:DLA_comm_relation} are deferred to Appendix~\ref{subsec:rep_sum_subrep} and Appendix ~\ref{subsec:DLA_comm_relation}.
\begin{lemma}\label{lem:rep_sum_subrep}
    Let $\mathfrak{g}$ be a dynamical Lie algebra and $r$ be the natural representation on the vector space $V$ satisfying $r(g)=g$ for any $g\in \mathfrak{g}$. If there exists irreducible subrepresentations of $r$ on $V$ such that $r(g)=r_1(g) \oplus \cdots \oplus r_k(g)$ acting on the space $V=V_1\oplus \cdots \oplus V_k$ for any $g\in \mathfrak{g}$, then the dimension of Lie algebra yields 
    \begin{equation}\label{eq:dim_sum_subrep}
\dim(\mathfrak{g})=\dim(r)=\sum_{j=1}^{k}\dim(r_j)=\sum_{j=1}^{k}\dim(V_j).
    \end{equation}
    where the dimension of subrepresentation $r_j$ refers to $\dim(r_j)=\dim(V_j)$.
\end{lemma}

\begin{lemma}[Commutant structure~\cite{simon1996representations}]
    Let $r$ be a representation of a Lie algebra $\mathfrak{g}$ on the Hilbert space $\mathcal{H}$ and its decomposition into irreducible representation be
    \begin{equation}
        r(g)=\oplus_{j=1}^k \mathbb{I}_{m_j} \otimes r_j(g),
    \end{equation}
    where $m_j$ is known as the multiplicity of the irreducible representation $r_j$. Then the elements of its commutant are of the following form
    \begin{equation}
        \mathcal{C}(\mathfrak{g})=\oplus_{j=1}^k \mathcal{C}_{j}(\mathfrak{g})\otimes \mathbb{I}_{\dim(m_j)}, 
    \end{equation}
    where $\mathcal{C}_{j}(\mathfrak{g})$ denotes bounded operators in a $m_j$-dimensional Hilbert space. Then the dimension of representation $r$ and subrepresentation $r_j$ yields
    \begin{equation}
        \dim(r)=\dim(\mathcal{C}(\mathfrak{g})), ~ \mbox{and} ~ \dim(r_j)=\dim(\mathcal{C}_{j}(\mathfrak{g}))
    \end{equation}
    \label{lem:comm_stru}
\end{lemma}

\begin{lemma}
    \label{lem:DLA_comm_relation}
    Let $\mathfrak{g}_{FG},\mathfrak{g}_{PG},\mathfrak{g}_{NG}$ be the Lie algebra related to the ansatz designs of the circuits with parameters fully grouping $\mathcal{A}_{FG}$, partially grouping $\mathcal{A}_{PG}$, no-grouping $\mathcal{A}_{NG}$. Then the related commutants of the three Lie algebras yield
    \begin{equation}\label{eq:DLA_comm_relation}
        \mathcal{C}(\mathfrak{g}_{NG})\subseteq\mathcal{C}(\mathfrak{g}_{FG}) = \mathcal{C}(\mathfrak{g}_{PG}),
    \end{equation}
    where the equality in the subset holds if there is no spatial symmetry in the problem Hamiltonian.
\end{lemma}
 
We now begin to present the proof of Lemma~\ref{lem:DLA_dim_relation}.
\begin{proof}
    [Proof of Theorem~\ref{lem:DLA_dim_relation}.] Following Lemma~\ref{lem:rep_sum_subrep} with denoting $r$ be the natural representation of $\mathfrak{g}$ on vector space $V$, the dimension of DLA $\mathfrak{g}$ is equal to the sum of dimensions of irreducible subrepresentations, i.e., 
    \begin{equation}
        \dim(\mathfrak{g})=\dim(r)=\sum_{j=1}^k\dim(r_j)=\sum_{j=1}^k\dim(V_j),
    \end{equation}
    where $V_j$ is the irreducible invariant subspace related to the subrepresentation $r_j$. For the symmetric ansatz design $\mathcal{A}$, there exsits an invariant space $V_*\in\{V_j\}_{j=1}^k$ such that the effective dimension $d_{\eff}(\mathcal{A})=\dim(V_*)$.

    To obtain Eqn.~(\ref{eq:DLA_dim_relation}), we first show that the effective dimension of DLA $\mathfrak{g}$ is inversely proportional to the size of commutant of the DLA $\mathfrak{g}$, and then show that the commutant sizes related to ansatz design $\mathcal{A}_{FG}, \mathcal{A}_{PG}, \mathcal{A}_{NG}$ are monotonically non-increasing. In particular, the commutant of Lie algebra $\mathfrak{g}$, denoted as $\mathcal{C}(\mathfrak{g})=\{V\in SU(d): [V,g]=0\}$, includes all the symmetry operator of the corresponding ansatz design. For any two Lie algebras $\mathfrak{g}_1, \mathfrak{g}_2$ with $\mathcal{C}(\mathfrak{g}_1)\subset \mathcal{C}(\mathfrak{g}_2)$, then any block diagonalization of the elements in $\mathcal{C}(\mathfrak{g}_1)$ is also the block diagonalization of the elements in $\mathcal{C}(\mathfrak{g}_2)$. This indicates that any invariant subspace of $\mathcal{C}(\mathfrak{g}_1)$ is also the invariant subspace of $\mathcal{C}(\mathfrak{g}_2)$, leading to $\dim(\mathcal{C}_{j}(\mathfrak{g}_2))\le \dim(\mathcal{C}_{j}(\mathfrak{g}_1))$. Following Lemma~\ref{lem:comm_stru}, we have
    \begin{equation}
        d_{\eff}(\mathfrak{g}_2)=\dim(r_*(\mathfrak{g}_2))=\dim(\mathcal{C}_{*}(\mathfrak{g}_2))\le \dim(\mathcal{C}_{*}(\mathfrak{g}_1))=\dim(r_*(\mathfrak{g}_1))=d_{\eff}(\mathfrak{g}_1), \label{eq:g1_less_g2}
    \end{equation}
    where $*\in [k]$ refers to the index of invariant space the optimization performs on and $\dim(r_*(\mathfrak{g}_j))$ with $j=1,2$ refers to the effective dimension related to the DLA $\mathfrak{g}_j$.
    In conjunction with Lemma~\ref{lem:DLA_comm_relation} and Eqn.~(\ref{eq:g1_less_g2}),
    we have $\mathcal{C}(\mathfrak{g}_{NG})\subseteq \mathcal{C}(\mathfrak{g}_{PG})=\mathcal{C}(\mathfrak{g}_{FG})$ and hence $d_{\eff}(\mathfrak{g}_{FG}) = d_{\eff}(\mathfrak{g}_{PG}) \le d_{\eff}(\mathfrak{g}_{NG})$. This completes the proof.
\end{proof}

\subsection{Proof of Lemma~\ref{lem:rep_sum_subrep}}\label{subsec:rep_sum_subrep}
\begin{proof}
    [Proof of Lemma~\ref{lem:rep_sum_subrep}]    The first equality in Eqn.~(\ref{eq:dim_sum_subrep}) follows the fact that natural representation $r$ is bijective and does not change the dimension of pre-image space. the second equality follows the definition of the dimension of representation in Definition~\ref{def:rep_LA} such that
    \begin{equation}
        \dim(r)=\dim(V)=\dim(V_1\oplus \cdots \oplus V_k)=\sum_{j=1}^k\dim(V_j)=\sum_{j=1}^k\dim(r_j),
    \end{equation}
    where the last equality follows that $r_j$ is a representation of $\mathfrak{g}$ on the space $V_j$. This completes the proof.
\end{proof}

\subsection{Proof of Lemma~\ref{lem:DLA_comm_relation}}\label{subsec:DLA_comm_relation}
\begin{proof}[Proof of Lemma~\ref{lem:DLA_comm_relation}]
     We begin this proof by showing that the commutant of $A=\{H_1\otimes \mathbb{I},\mathbb{I} \otimes H_2\}$ is a subset of $B=\{H_1\otimes \mathbb{I}+ \mathbb{I}\otimes H_2\}$, where $H_1, H_2$ are arbitrary Hermitian operators and $B$ refers to the set with imposing parameter grouping on $A$. In particular, for any matrix $S$ which commutes with the elements in $A$, we have
     \begin{equation}
         S(H_1\otimes \mathbb{I}+ \mathbb{I}\otimes H_2)=S (H_1\otimes \mathbb{I})+ S(\mathbb{I}\otimes H_2)= (H_1\otimes \mathbb{I})S+ (\mathbb{I}\otimes H_2)S=(H_1\otimes \mathbb{I}+ \mathbb{I}\otimes H_2)S.
     \end{equation}
     This indicates that $\mathcal{C}(A)\subseteq \mathcal{C}(B)$.
     With this fact, we now derive the Eqn.~(\ref{eq:DLA_comm_relation}). We first recall that the generators of the Lie algebras $\mathfrak{g}_{NG},\mathfrak{g}_{PG},\mathfrak{g}_{FG}$ yield non-discreasingly restrictive parameters grouping strategy, and are identity when there is no spatial symmetry in the problem Hamiltonian, i.e., $\mathfrak{g}_{FG}=\mathfrak{g}_{PG}=\mathfrak{g}_{NG}$. Moreover, the definition of $\mathfrak{g}_{FG},\mathfrak{g}_{PG}$ indicates that the related ansatzes follow the same symmetry, namely, any unitary $U$ commutes with the elements in  $\mathfrak{g}_{FG}$ if and only if $U$ commutes with the elements in  $\mathfrak{g}_{PG}$. Hence we have $\mathcal{C}(\mathfrak{g}_{FG}) = \mathcal{C}(\mathfrak{g}_{PG})$ as the commutant consists of the symmetry operator of the ansatz design.

     On the other hand, the relation $\mathcal{C}(\mathfrak{g}_{PG}) \subseteq \mathcal{C}(\mathfrak{g}_{NG})$ in Eqn.~(\ref{eq:DLA_comm_relation}) directly following the analog between the set $A$ and $B$ and the generators related to the Lie algebra $\mathfrak{g}_{PG}$ and $\mathfrak{g}_{NG}$, where the generators related to $\mathfrak{g}_{PG}$ refers to the set with imposing parameters grouping on $\mathfrak{g}_{NG}$. This completes the proof.
\end{proof}

%%%%%%%%%%%%%%%%%%%%%%%%%%%%%%%%%%%%%%%%%%%%%%%%%%%%%%%%%%%%%%%%%%%%%%%%%%%%%%%
%%%%%%%%%%%%%%%%%%%%%%%%%%%%%%%%%%%%%%%%%%%%%%%%%%%%%%%%%%%%%%%%%%%%%%%%%%%%%%%

\section{Related work}\label{app:related_work}

In this section, we embark on a concise literature review, focusing on conventional algorithms for the Max-Cut problem, some variants of QAOA, and quantum circuit architecture search algorithms. This examination sets the stage for a comparative analysis between these established methods and our proposed model. In summary, our discussion underscores the distinctive strength of our model: its exceptional ability to generalize.

\subsection{Conventional algorithms}

\textbf{Greedy algorithm for Max-Cut problem.} The greedy algorithm for solving the Max-Cut problem operates on a simple principle: iteratively makes local, myopic decisions to construct a solution that attempts to maximize the sum of weights of edges between two disjoint subsets of vertices. This algorithm does not assure an optimal solution due to its greedy nature—making decisions based only on immediate benefits without considering future consequences. The detailed procedure is introduced in Alg.~\ref{alg:greedy}.

\begin{algorithm}
\caption{\textbf{Greedy Algorithm for weighted Max-Cut}}
\label{alg:greedy}
\begin{algorithmic}[1]
\STATE \textbf{Input:} A graph $G=(V,E)$ with weights $w_{ij}$ on edges $(i,j) \in E$
\STATE \textbf{Output:} A partition of $V$ into subsets $S$ and $\bar{S}$ maximizing the cut weight

\STATE Initialize $S = \emptyset$, $\bar{S} = V$
\STATE Initialize $cutWeight = 0$
\FOR{each vertex $v \in V$}
    \STATE $deltaWeight = 0$
    \FOR{each edge $(v, u) \in E$ connected to $v$}
        \IF{$u \in S$ and $v \notin S$ or $u \notin S$ and $v \in S$}
            \STATE $deltaWeight = deltaWeight - w(v,u)$
        \ELSE
            \STATE $deltaWeight = deltaWeight + w(v,u)$
        \ENDIF
    \ENDFOR
    \IF{$deltaWeight > 0$}
        \IF{$v \in S$}
            \STATE Move $v$ to $\bar{S}$ and update $cutWeight += deltaWeight$
        \ELSE
            \STATE Move $v$ to $S$ and update $cutWeight += deltaWeight$
        \ENDIF
    \ENDIF
\ENDFOR
\STATE \textbf{return} $S$, $\bar{S}$, $cutWeight$
\end{algorithmic}
\end{algorithm}

\textbf{Goemans-Williamson (GW) algorithm for Max-Cut problem.} The GW algorithm utilizes semidefinite programming to relax the original combinatorial problem into a continuous one that can be solved efficiently. After solving the semidefinite program, the algorithm uses a random hyperplane to split the vertices into two subsets, which form the cut. The GW algorithm achieves an approximation ratio of at least $0.878$ for the Max-Cut problem. The simplified pseudocode of GW algorithm is described in Alg.~\ref{alg:gw}.

\begin{algorithm}
\caption{\textbf{Goemans-Williamson Algorithm for Max-Cut}}
\label{alg:gw}
\begin{algorithmic}[1]
\STATE \textbf{Input:} A graph $G=(V,E)$ with weights $w_{ij}$ on edges $(i,j) \in E$
\STATE \textbf{Output:} A partition of $V$ into subsets $S$ and $\bar{S}$

\STATE Formulate the Max-Cut problem as a semidefinite programming (SDP) problem.
\STATE Solve the SDP problem to find a vector representation $\vec{v}_i$ for each vertex $i$.
\STATE Choose a random hyperplane by selecting a random unit vector $\vec{r}$.
\FOR{each vertex $i \in V$}
    \IF{$\vec{v}_i \cdot \vec{r} \geq 0$}
        \STATE Assign vertex $i$ to subset $S$
    \ELSE
        \STATE Assign vertex $i$ to subset $\bar{S}$
    \ENDIF
\ENDFOR
\STATE \textbf{return} $S$, $\bar{S}$
\end{algorithmic}
\end{algorithm}

\subsection{Variants of QAOA}
The studies of variants of QAOA aim to improve the convergence rate or reduce the computational time by changing the PQCs or the problem Hamiltonian. Current progress has revealed that the performance of QAOA could be improved by employing multi-angle QAOA \cite {herrman2022multi} where the parameters are no-grouped or partially grouped according to the permutation symmetry of problem Hamiltonian \cite{shaydulin2021exploiting,shi2022multiangle,sauvage2022building}, utilizing different mixer Hamiltonian obtained by searching from a given Hamiltonian pool \cite{zhu2022adaptive} or inspired by specific problem \cite{chalupnik2022augmenting,yu2022quantum,hadfield2019quantum,yoshioka2023fermionic} and other quantum algorithms \cite{chandarana2022digitized,wurtz2022counterdiabaticity,bartschi2020grover}. Another type of the variant of QAOA focuses on modifying the problem Hamiltonian, either through eliminating redundant qubits \cite{bravyi2020obstacles} to obtain a reduced problem Hamiltonian, or imposing conditional rotations \cite{villalba2021improvement} to the Hamiltonian. In the following, we delve into the most relevant variants of QAOA to our study and compare them with our model.

\textbf{Multi-Angle QAOA (ma-QAOA) \cite{herrman2022multi}.} The ma-QAOA innovates on the traditional QAOA framework by incorporating a larger set of parameters. It allows each operator within both the cost and mixer Hamiltonians to be governed by its own unique parameter, diverging from the conventional approach where a single parameter is shared among all operators. In our experiment, attention is focused exclusively on the modifications within the mixer Hamiltonian for fair comparison. The new mixer Hamiltonian is expressed as
\begin{equation}
    H_M=\sum_{i=1}^N\beta_iX_i.
\end{equation}
where $X_i$ denotes the Pauli-X operation applied to the $i$-th qubit and $\beta_i$ represents the corresponding individual parameter. This adjustment significantly expands the parameter space in ma-QAOA, scaling the total count from $2p$ in the standard QAOA to $(N+1)p$. Despite empirical evidence suggesting that ma-QAOA surpasses the original QAOA in achieving higher approximation ratios for configurations with fewer layers, the complexity introduced by the augmented parameter space could potentially impede its effectiveness in scenarios involving deeper circuits.

\textbf{ADAPT-QAOA \cite{zhu2022adaptive}.} In ADAPT-QAOA, the mixer Hamiltonian is selected from a pre-defined operator pool $\{A_j\}$ step by step. For the $k$-step, the operators $A_j$ is guided by maximizing the following gradient:
\begin{equation}
    -i\braket{\psi_{k-1}(\bm{\alpha},\bm{\beta})|e^{i\alpha_kH_C[H_C,A_j]e^{-i\alpha_kH_C}|\psi_{k-1}(\bm{\alpha},\bm{\beta})}},
\end{equation}
where $\ket{\psi_{p}(\bm{\alpha},\bm{\beta})}=(\prod_{k=1}^p e^{-i\beta_k A_k}e^{-i\alpha_k H_C})\ket{\psi_0}$. Following the selection of $A_j$, all parameters undergo a subsequent optimization phase. This procedure is iterated until the gradient's norm falls below a set threshold, or the circuit reaches its predefined maximum depth. ADAPT-QAOA's dynamic mixer Hamiltonian selection aims to potentially discover a more direct path to adiabaticity, thereby enabling accelerated convergence. However, its practicality for large-scale problems is hampered by the increased measurement costs required for gradient evaluation, a factor contingent on the size of the operator pool.

Contrasting with these QAOA variants, MG-Net uniquely offers a dynamic offline adaptation of the mixer Hamiltonian, tailoring it to the specific problem and circuit depth without incurring extra computational costs. Additionally, MG-Net demonstrates remarkable generalization capabilities, effectively learning from a limited dataset to address a broad spectrum of problems. This facilitates the rapid development of mixer Hamiltonians for new problems.

\subsection{Quantum circuit architecture search}

In the design of quantum circuits, quantum circuit architecture search methodologies have been developed to autonomously identify optimal quantum circuit architectures \cite{zhang2021neural,ye2021quantum,ostaszewski2021reinforcement,kuo2021quantum,meng2021quantum,du2022quantum,linghu2022quantum,he2022quantum,zhang2022differentiable,wu2023quantumdarts,lei2024neural,lu2023qas}. In the following, we delve into several notable approaches and contrast them with our MG-Net model.

\textbf{Quantum architecture search (QAS) \cite{du2022quantum}.} The QAS approach automatically seeks an optimal quantum circuit architecture to balance the benefits and side effects of adding more quantum gates, considering the noise in quantum systems. This method involves several steps: initializing a superstructure (supernet) that defines the pool of potential architectures, optimizing parameters across these architectures, ranking them based on performance, and finally refining the chosen architecture.

\textbf{Differentiable Quantum Architecture Search (DQAS) \cite{zhang2022differentiable}}. DQAS introduces a novel approach by employing differentiable programming techniques. This method enables the concurrent optimization of both the structure and parameters of quantum circuits through gradient descent, streamlining the search process.

\textbf{QuantumDARTS \cite{wu2023quantumdarts}.} The QuantumDARTS algorithm, which leverages the Gumbel-Softmax technique for differential optimization of quantum circuit structure and parameters, aims to reduce the search cost by following two search strategies: macro search for entire circuit optimization and micro search for sub-circuit structures, improving its adaptability to large-scale problems.

Despite their advancements, these QAS methodologies share a fundamental limitation: they are inherently designed to address singular, specific problems. Consequently, adapting these methods to new problems necessitates repeating the resource-intensive architecture search process from scratch. In contrast, MG-Net exhibits an unparalleled ability to generalize across a spectrum of problems based on a minimal set of training examples. This capability enables MG-Net to rapidly design optimal circuits for novel problems through a single feedforward computation, bypassing the need for repeated, exhaustive searches. This unique advantage positions MG-Net as a highly efficient and versatile tool in the quantum computing landscape, offering significant savings in computational resources and time.

%%%%%%%%%%%%%%%%%%%%%%%%%%%%%%%%%%%%%%%%%%%%%%%%%%%%%%%%%%%%%%%%%%%%%%%%%%%%%%
%%%%%%%%%%%%%%%%%%%%%%%%%%%%%%%%%%%%%%%%%%%%%%%%%%%%%%%%%%%%%%%%%%%%%%%%%%%%%%

\section{Implementation details of MG-Net}\label{app:implementation}

In this section, we initially outline the methodology for constructing datasets used to train MG-Net across various problem scales. Subsequently, we detail the implementation of the data encoder, illustrated with a specific example.

\subsection{Dataset construction}\label{app:implement_detail_dc}

\textbf{Operator types.} The set of operator types for the mixer Hamiltonian is defined as $\{X, Y\}^{\otimes N}$ in our experiments. Note that the operator type pool can be flexibly adjusted according to specific problems and hardware. For example, we can introduce two-qubit operators into
the operator type pool to further enhance the performance of QAOA, as done in \cite{zhu2020adaptive}. Considering the exponential growth of the search space in relation to the system size $N$, we have sampled only a subset from this pool in all our experiments. This approach is adopted to construct the training dataset while minimizing data collection costs.

\textbf{Construction of parameter group pool.} A straightforward idea to construct the pool of parameter group is to assume each $X_i$ can be assigned an index $j$ ranging from $1$ to $N$, leading to a pool $P=\{(j_1\in [N],...,j_N\in [N])\}$ with size $N^N$. However, there exist multiple duplicate candidates in the pool $P$ due to the disorder of the initial parameter pool. For example, for a two-qubit QAOA ansatz, parameter index vectors $(1,2)$ and $(2,1)$ make no difference in the optimization of QAOA. Based on these observations, we propose a recursive algorithm Alg.~\ref{alg:pool-construct} to build a compact pool of parameter groups.

\begin{algorithm}[H]
\caption{\small{\textbf{Construction of parameter group pool}}}
\label{alg:pool-construct}
   \begin{algorithmic}[1]
   \STATE \textbf{Input}: The qubit number $N$, pool $P=\{\}$ 
   \STATE \textbf{Output}: Pool $P$
   \STATE \textbf{Function} grouping\_pool($max\_index, index\_list, N$) 
        \STATE \hspace*{1em} \textbf{if} $\text{length}(index\_list) == N$
            \STATE \hspace*{2em} Add $index\_list$ to $P$
            \STATE \hspace*{2em} \textbf{return}
        \STATE \hspace*{1em} \textbf{end if}
        \STATE \hspace*{1em} \textbf{for} {$i = 1,\cdots,max\_index$}
            \STATE \hspace*{2em} Append $i$ to $index\_list$
            \STATE \hspace*{2em} grouping\_pool($\max(max\_index, index\_list[-1] + 2), index\_list, N$)
            \STATE \hspace*{2em} Delete the last element of $index\_list$
        \STATE \hspace*{1em} \textbf{end for}
   \STATE \textbf{End Function}
   \STATE grouping\_pool($2, \text{empty\_list}, N$)
\end{algorithmic}
\end{algorithm}

In practice, we randomly selected $5$ candidates from the parameter grouping pool for each operator type. Although the training dataset only partially covers the entire space of operator types and parameter groupings, our model is still capable of learning the intrinsic relationship between the mixer Hamiltonian and its corresponding achievable cost.

To find the minimal cost that can be achieved by a QAOA
circuit during the construction of the training dataset in stage 1, we run the same QAOA
circuit 10 times and record their cost values. For each run, the QAOA circuit
is initialized with different random parameters and optimized for 40 epochs. Finally,
the minimum of these cost values is selected as the label that represents the minimal
achievable cost.

\textbf{Large-scale dataset.} To assess our method's efficacy on large-scale problems, we concentrated on the Max-Cut problem using weighted graphs with $64$ nodes. Simulating larger-scale quantum circuits on classical devices poses significant challenges. To overcome this, our approach employs a divide-and-conquer strategy, simulating a large-scale circuit through multiple smaller-scale circuits. We then integrate the results of these smaller circuits to estimate the performance of the original large-scale circuit. For a detailed explanation of this methodology, refer to QAOA-in-QAOA \cite{zhou2023qaoa}.

In constructing the training dataset $D_{\rm ce}^{\rm Tr}$ for $64$-node graphs, we divide each $64$-node graph into $8$ sub-graphs, each containing $8$ nodes. The max-cut of each sub-graph is computed using an $8$-qubit QAOA. To gather a comprehensive range of samples, we vary the operator types and parameter groupings in the $8$-qubit circuits, which in turn simulates the variation in mixer Hamiltonians for $64$-qubit circuits. It is important to note that these $8$-qubit circuits operate independently, with no shared parameters, resulting in at least $8$ independent parameters for each $64$-qubit circuit in our training dataset. For testing on the unknown graphs, we employ tensor network simulations to accurately estimate the performance of the original $64$-qubit QAOA.

\begin{figure*}[htp]
    \centering
    \includegraphics[width=0.9\linewidth]{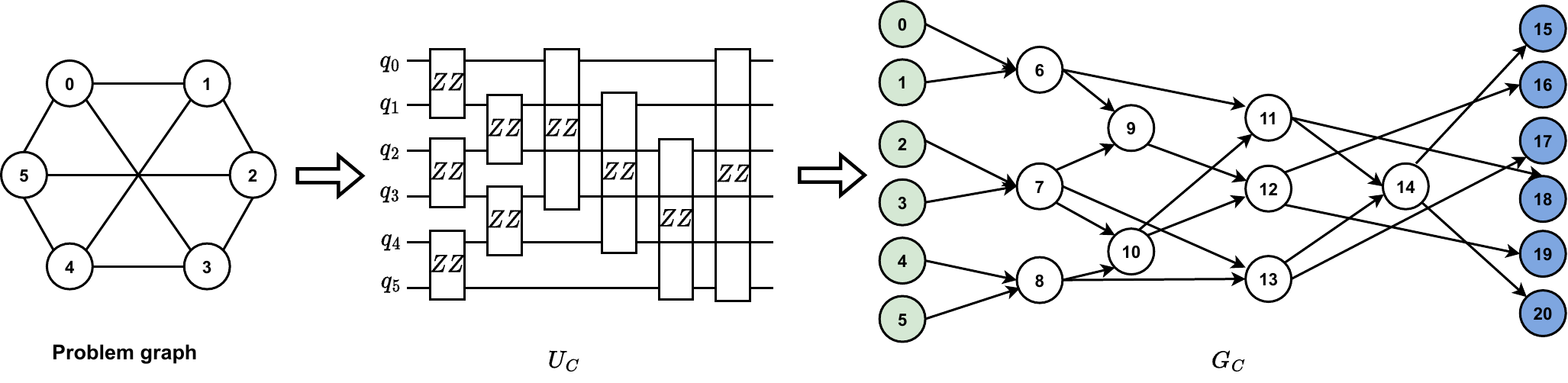}
    \caption{\small{\textbf{Encoding of problem.} The problem graph is first transformed into a quantum circuit, which is subsequently encoded by a DAG.}}
    \label{fig:encode_pg}
\end{figure*}

\subsection{Data encoder}\label{app:implement_detail_de}

\textbf{Problem encoder.} Our problem encoder is rooted on the problem Hamiltonian $H_C$ in Eqn.~(\ref{eq:para_state}). More precisely, to facilitate a consistent and unified representation for diverse combinatorial problems $\{G\}$, we initiate by converting the original problem $G$ into the corresponding unitary $U_C=\exp(-i\alpha_kH_C)$, which is subsequently transformed into a directed acyclic graph (DAG) $G_C$.

Fig.~\ref{fig:encode_pg} illustrates the problem encoding process for a regular graph with $6$ nodes. Each node of the problem graph corresponds to a qubit in the quantum system and each edge $(i,j)$ is represented as a two-qubit gate $Z_iZ_j$, which is exactly the problem Hamiltonian of QAOA for the Max-Cut problem. Based on this problem unitary, we construct the final graph representation $G_C$, with each two-qubit gate depicted as a node in the graph. In addition to these gate-induced nodes, two unique node types, the input and output nodes which correspond to qubits, are introduced to denote the start and end of $G_C$, respectively. The edges of $G_C$ signify the temporal order of quantum gate execution, linking consecutive gates and thereby dictating the flow of the quantum computation. The weights of edges are encoded into the node feature.

\begin{figure}[htp]
    \centering
    \includegraphics[width=0.4\linewidth]{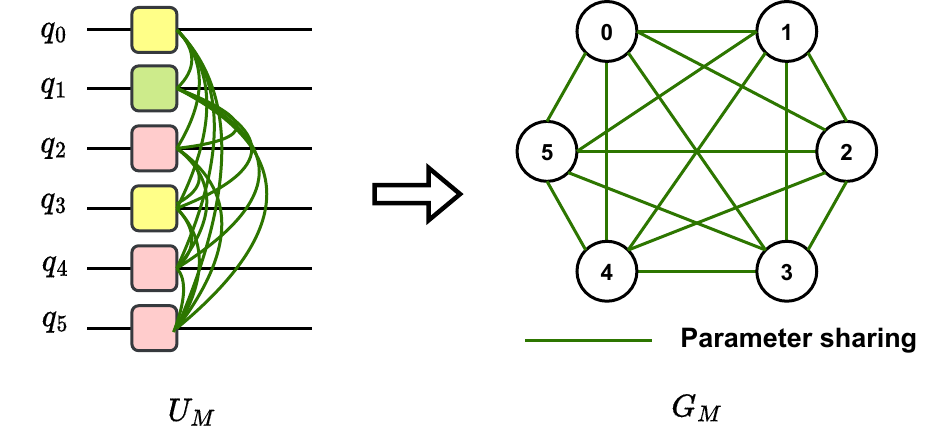}
    \caption{\small{\textbf{Encoding of mixer Hamiltonian.} Each qubit in the mixer Hamiltonian is represented as a node in the encoded graph. The type of operator associated with each qubit is encoded in the node feature, while the parameter grouping strategy is encapsulated in the edge features.}}
    \label{fig:encode_mg}
\end{figure}

\textbf{Mixer encoder.} We define a one-to-one mapping to encode the candidate mixer Hamiltonian $H_M$ as a graph $G_M$. Recall Eqn.~(\ref{eq:hm}), two types of information about $H_M$ should be encoded in $G_M$ are operators $\{P_i\}$ and the parameter grouping strategy $\mathcal{G}$. In MG-Net, each operator is modeled as a node of $G_M$, and the operator type is encoded as part of the node feature vector. Concretely, MG-Net initially constructs $G_M$ as a fully connected graph, where the edge weight is a binary variable, representing whether the two operators connected by the edge share the same control parameter.

The process of encoding a mixer Hamiltonian into a graph representation is illustrated in Fig.~\ref{fig:encode_mg}. Here, we take the example of a $6$-qubit mixer Hamiltonian encoded as graph $G_M$. In this graph, each qubit's corresponding operator is depicted as a node, with the operator acting on the $i$-th qubit represented by the $i$-th node in $G_M$. The graph's edges signify the parameter correlations among these operators. Specifically, let $w_{ij}\in \{0,1\}$ be the weight of edge connecting node $i$ and $j$. If the operator $i$ and $j$ share the same parameter, then $w_{ij}=0$; otherwise, $w_{ij}=1$. 

\textbf{Depth embedding.} The circuit depth $p$ is encoded as a vector $\bm{x}_p$ through position embedding \cite{vaswani2017attention}. Mathematically, $\bm{x}_p$ is constructed as
\begin{equation}
\begin{aligned}
    \bm{x}_p[2k]=\sin{\frac{p}{10000^{2k/d_p}}},\bm{x}_p[2k+1]=\cos{\frac{p}{10000^{2k/d_p}}},\nonumber
\end{aligned}
\end{equation}
where $d_p$ is dimension of $\bm{x}_p$ and $k=0, ..., \lfloor d_p/2 \rfloor$.

\subsection{Network structure}\label{app:net_arch}

\subsubsection{Cost estimator}

In our experimental setup, the intricate architecture of the cost estimator is detailed in Fig.~\ref{fig:ce_detail}. Both the problem and mixer Hamiltonian branches incorporate two layers of graph convolutions, utilizing ReLU activation functions to transform the initial node features from dimensions $d_C$ and $d_M$ to a unified $128$-dimensional space. Subsequently, the three extracted features—$\bm{x}_C$, $\bm{x}_M$, and $\bm{x}_p$—are concatenated to facilitate the prediction of the attainable minimum cost $\hat{y}$ for a given QAOA instance through an MLP layer.

\begin{figure}[htp]
    \centering
    \includegraphics[width=0.9\linewidth]{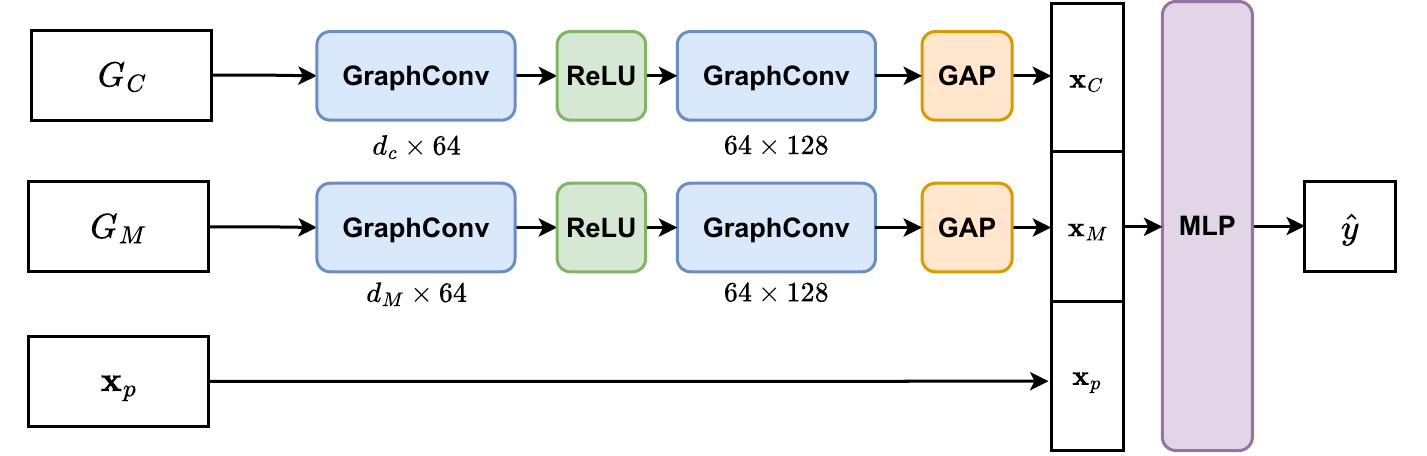}
    \caption{\small{\textbf{Implementation of cost estimator.} The term `GraphConv' represents the graph convolution module. `ReLU' is a commonly used activation function in neural networks. $d_C$ and $d_M$ represent the dimension of node feature in graph $G_C$ and $G_M$ respectively. $P_i$ represents the operator type for the $i$-qubit and $\bm{e}_{ij}$ represents the weight for edge $(i,j)$.}}
    \label{fig:ce_detail}
\end{figure}

\subsubsection{Mixer generator}

Inspired by \cite{qian2024multimodal} which encodes a quantum circuit as a graph, the mixer generation is composed of two separate sub-generators: the operator type generator and the parameter grouping generator, which are respectively responsible for graph node and link prediction.

\textbf{Operator type generator.} The task of generating operator types $\mathcal{P}$ is conceptualized as a graph node classification task. Specifically, we employ a GNN to process $G_C$, identifying output nodes to represent the operators corresponding to each qubit, while disregarding irrelevant nodes. To incorporate the circuit depth $p$ into the prediction, we enhance the feature set of each output node by appending a feature vector $\bm{x}_p$. This enriched node feature set is then fed into an MLP to predict the specific category of each operator.

\textbf{Parameter grouping generator.} Recall that the grouping strategy is traditionally represented by sets of index groups $\{\mathcal{G}_j\}_{j=1}^K$ with an unspecified $K$, posing a challenge for neural network processing. To address this, we extend the parameter grouping problem as follows: if an edge indicator $\bm{e}_{ij}=1$, then the mixer operators $P_i$ and $P_j$ are correlated and share the same parameter; otherwise, they are controlled by independent parameters. Furthermore, if $\bm{e}_{ij}=1$ and $\bm{e}_{ik,k\neq j}=1$, then $P_i$, $P_j$ and $P_k$ are correlated regardless of the value of $\bm{e}_{jk}$. In this way, the parameter grouping task is translated into the prediction of the binary variable $\bm{e}_{ij}\in \{0,1\}$, as a link prediction task. This modeling bypasses the need to predetermine the number of parameter groups and offers flexibility in incorporating constraints related to qubit connections.

Analogous to the operator type generator, the parameter grouping generator employs another GNN to process $G_C$ to extract features of output nodes, which are then extended with circuit depth feature $\bm{x}_p$. For node $i$ and $j$, their extended features $\bm{x}_i$ and $\bm{x}_j$ are used to determine the existence of an edge $(i,j)$ by evaluating $\bm{e}_{ij}=B({\rm MLP}(\bm{x}_i \circ \bm{x}_j))$, where $B(\cdot)$ signifies a binarization function. In MG-Net, this function is realized using the Gumbel-Softmax trick, ensuring the differentiability.

In our experiment, the detailed structure of the mixer generator is depicted in Fig.~\ref{fig:mg_detail}. The mixer generator integrates two specialized branches to analyze the input problem graph $G_C$, with each branch deploying two graph convolution layers to distill the feature vector $\bm{x}_C$ with a dimensionality of $128$. This feature vector is then augmented with the circuit depth feature $\bm{x}_p$ to enrich the predictive capability of the model. For the precise prediction of operator types $\{P_i\}_{i=1}^N$ applicable to each qubit, the terminal nodes of $G_C$ are chosen for input into a Multi-Layer Perceptron (MLP) layer. This step calculates the likelihood of each potential operator type. Concurrently, a separate MLP layer is employed to ascertain the parameter sharing between operators $P_i$ and $P_j$. This is achieved through the equation $\bm{e}{ij}={\rm MLP}(\bm{x}_i \circ \bm{x}_j)$, where $\circ$ denotes the element-wise multiplication, and $\bm{x}_i$ symbolizes the enriched feature of the $i$-th node. 

\begin{figure}[htp]
    \centering
    \includegraphics[width=0.9\linewidth]{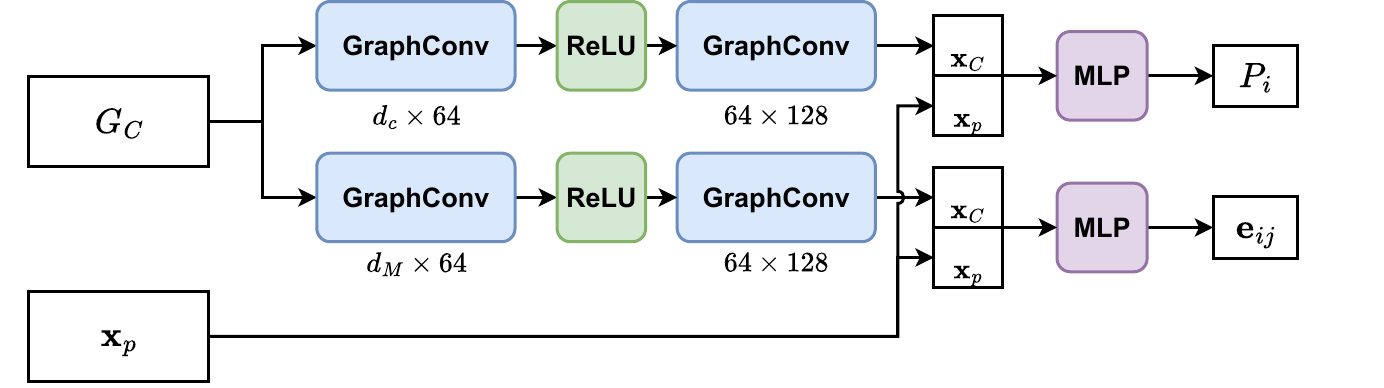}
    \caption{\small{\textbf{Implementation of mixer generator.} The term `GraphConv' represents the graph convolution module. `ReLU' is a commonly used activation function in neural networks. $d_C$ and $d_M$ represent the dimension of node feature in graph $G_C$ and $G_M$ respectively.}}
    \label{fig:mg_detail}
\end{figure}

\subsection{Experiment settings}

\textbf{Hardware platform.} All QAOA circuits are implemented by PennyLane \cite{bergholm2018pennylane} and run on classical device with Intel(R) Xeon(R) Gold 6267C CPU @ 2.60GHz and 128 GB memory. MG-Net is implemented by Pytorch \cite{paszke2019pytorch} and is trained on a single NVIDIA GeForce RT 2080Ti with 12G graphics memory.

\textbf{Hyper-parameters.} The hyper-parameters of optimizing MG-Net and QAOA circuit are listed in Tab.~\ref{tab:hyper-param}.

\textbf{Initial state.} The initial quantum state of the
QAOA circuit is consistently set to $\ket{+}^{\otimes N}$, irrespective of the mixer Hamiltonian chosen.
Although this approach does not ensure that the initial state is always the ground state
of the predicted mixer Hamiltonian, it does not compromise the QAOA’s performance
and has the potential to outperform the traditional state initialization technique, which can be partially explained by the physical intuition of counterdiabatic (CD) driving \cite{chandarana2022digitized,zhu2022adaptive}.

\begin{table}[htp]
    \centering
    \caption{\small{\textbf{The hyper-parameters of optimizing MG-Net and QAOA circuit.}}}
    \begin{tabular}{lcc}
        \toprule[1pt]
         & QAOA & MG-Net  \\
        \midrule
        optimizer & Adam & Adam \\
        learning rate & 0.15 & $1*10^{-4}$ \\
        epoch & 40 & 250 \\
        $\lambda_e$ & - & 1.0 \\
        $\lambda_r$ & - & 1.0 \\
        \bottomrule[1pt]
    \end{tabular}
    \label{tab:hyper-param}
\end{table}

\section{More numerical results}\label{app:more_result}

In this section, we initially show the results of comparing the approximation ratio achieved by different methods for TFIM. Then we examine how the approximation ratio achieved by various methods varies with different circuit depths $p$. Subsequently, we explore the convergence behavior of the QAOA when enhanced by our approach.

\subsection{Performance comparison among different methods for TFIM}\label{app:comp_tfim}

In evaluating the effectiveness of our proposed method for solving TFIM, we conducted a comparative analysis against QAOA, ADAPT-QAOA, and multi-angle QAOA (ma-QAOA). Our analysis, based on the average results from $100$ graphs in our test dataset, is summarized in Tab.~\ref{tab:tfim-ar}. The findings reveal that our method consistently outperforms other techniques in achieving a higher approximation ratio for TFIM, particularly in larger-scale problems.

\begin{table}[htp]
    \centering
    \caption{\small{\textbf{Comparison of approximation ratio $r$ among different methods for TFIM.}}}
    \begin{tabular}{lcc}
        \toprule[1pt]
        Method & $6$ qubits & $16$ qubits  \\
        \midrule
        QAOA & $0.990\pm  0.005$ & $0.523\pm 0.083$ \\
        ADAPT-QAOA & $0.857\pm 0.245$ & $0.742\pm 0.356$ \\
        ma-QAOA & $0.994\pm 0.001$ & $0.921\pm 0.040$ \\
        \midrule
        \textbf{Ours} & $\bm{0.996\pm 0.001}$ & $\bm{0.963\pm 0.031}$ \\
        \bottomrule[1pt]
    \end{tabular}
    \label{tab:tfim-ar}
\end{table}

\subsection{Experiments on asymmetric graphs and 2D-TFIM}

We conducted additional experiments on the asymmetric graphs of 6 nodes and 2D lattice models of $6$ spins. Their topological structure is shown in Fig.~\ref{fig:asym-2d}. 

\begin{figure}[htp]
    \centering
    \includegraphics[width=0.9\linewidth]{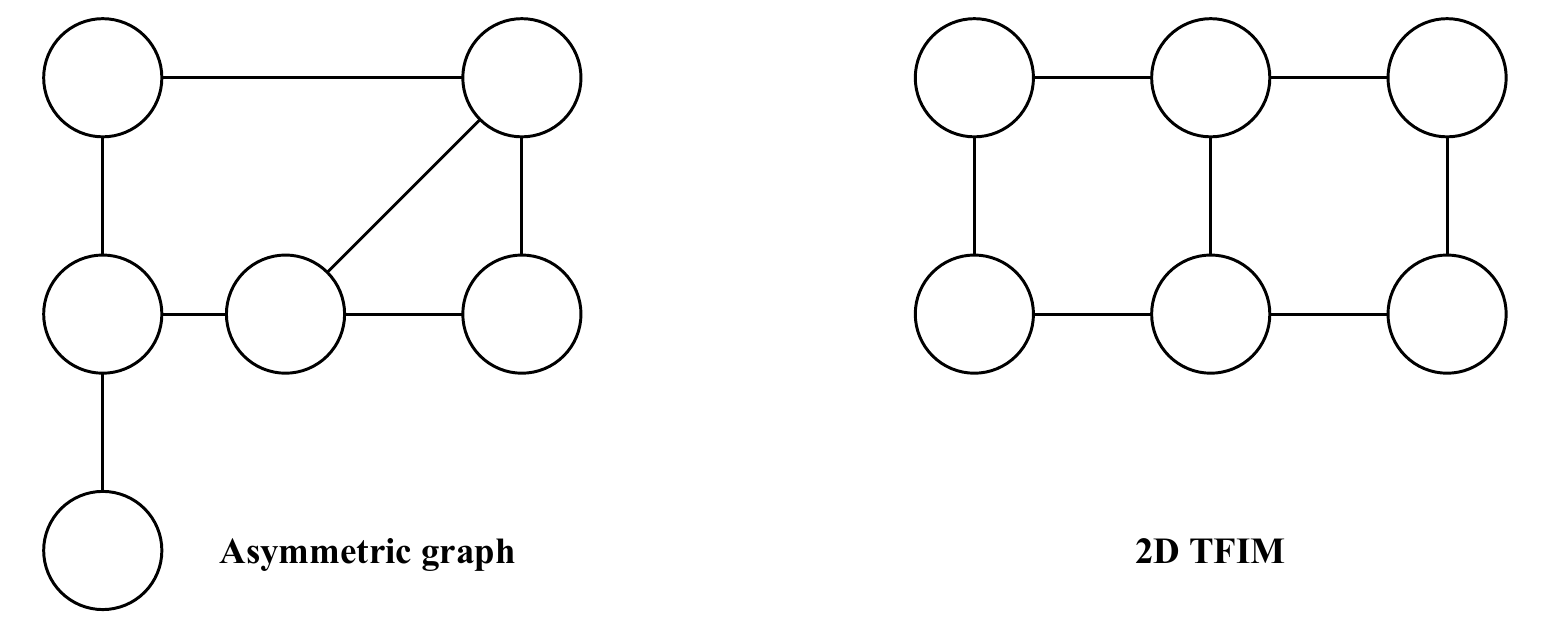}
    \caption{\small{\textbf{Topological structure of  asymmetric graphs and 2D TFIM.} }}
    \label{fig:asym-2d}
\end{figure}

The comparison of the achieved approximation ratio at $p=42$ over $100$ random test samples is summarized in the Tab.~\ref{tab:asym-2d}. The result affirms that our model consistently outperforms both standard
QAOA and ma-QAOA in terms of approximation ratio on more general cases.

\begin{table}[htp]
    \centering
    \begin{tabular}{lcc}
        \toprule[1pt]
        Tasks &  Max-Cut for asymmetric graphs & 2D TFIM  \\
        \midrule
        QAOA & $0.952\pm 0.026$ & $0.977\pm 0.008$   \\
        ma-QAOA & $0.987\pm 0.008$ & $0.980\pm 0.019$ \\
        \midrule
        \textbf{Ours} & $\bm{0.988\pm 0.005}$ & $\bm{0.988\pm 0.006}$  \\
        \bottomrule[1pt]
    \end{tabular}
    \label{tab:asym-2d}
\end{table}

% \subsection{Experiments on constrained combinatorial optimization problems}

\subsection{Approximation ratio with respect to $p$}

In small-scale quantum systems, achieving the criteria set in Theorem~\ref{thm:main_convergence} is more straightforward by increasing circuit depth $p$ beyond the threshold $C$. We analyze the approximation ratios achieved by $6$-qubit QAOA circuits for Max-Cut and TFIM within the $p$ range of $2$ to $82$. Figure~\ref{fig:ar_vs_p_n6} illustrates that at lower $p$ values, our method consistently records the highest approximation ratio $r$, clearly outperforming both standard QAOA and ma-QAOA. As $p$ increases from $2$ to $62$, standard QAOA and ma-QAOA exhibit a rise in $r$, eventually matching our method’s performance. However, a further increase in $p$ leads to a performance decline in ma-QAOA, where the detrimental impact of its numerous trainable parameters on convergence outweighs the benefits of enhanced expressibility. In contrast, our method maintains stable performance, continually achieving the highest $r$. These findings confirm our method's superiority in optimizing approximation ratios across various circuit depths compared to other approaches.

\begin{figure}[htp]
    \centering
    \includegraphics[width=0.9\linewidth]{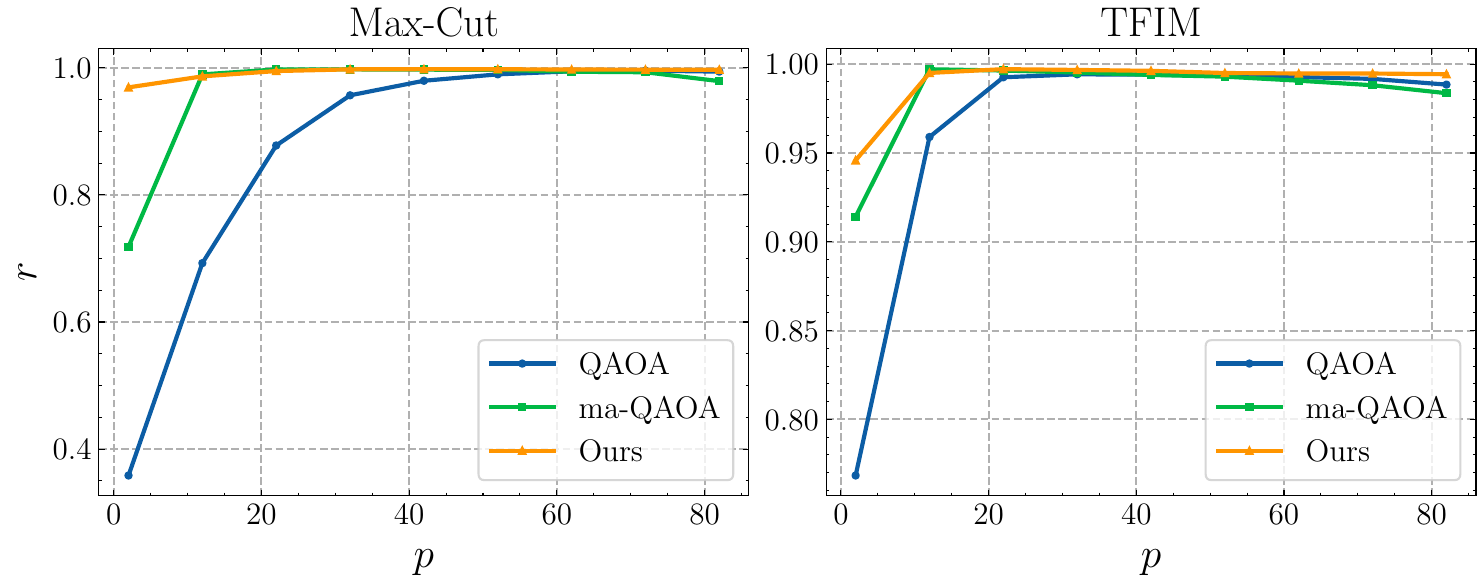}
    \caption{\small{\textbf{Comparison of the approximation ratio achieved by $6$-qubit QAOA, ma-QAOA and our model for Max-Cut and TFIM with varying $p$.} }}
    \label{fig:ar_vs_p_n6}
\end{figure}

We further explore the specific configurations of mixer Hamiltonians generated by MG-Net. Table~\ref{tab:mh} presents examples of predicted mixer Hamiltonians for $p$ values of ${12, 52, 82}$. At a smaller circuit depth of $p=12$, the optimal parameter grouping strategy maximizes the number of parameters, assigning each operator its independent parameter. This approach enhances the expressivity of the QAOA circuit and, alongside the introduction of novel mixer operators, contributes to superior approximation performance. For $p=52$, which verges on the threshold of over-parameterization, a trend towards grouping some operators is observed. At a higher circuit depth, such as $p=82$, the majority of operators are assigned the same parameter, aligning closer to the configuration of a standard QAOA circuit. The evolution of the mixer Hamiltonian configuration with varying $p$ partially reveals the underlying design principle of mixer Hamiltonian across different problems and circuit depths.
 
\begin{table*}[htp]
    \centering
    \caption{\small{\textbf{Operator type and parameter group generated by MG-Net.} `X' and `Y' represent Pauli-X and Pauli-Y, respectively. Parameter groups are formatted as $a_1-a_2-\cdots-a_N$, with $a_i \in \{0,1,...,N-1\}$ indicating the parameter index for the $i$-th operator. Identical indices ($a_i=a_j$) imply shared parameters between operators.}}
    \begin{tabular}{lccc}
        \toprule[1pt]
        Task & & Max-Cut & TFIM \\
        \midrule
        \multirow{2}{*}{$p=12$} & Operator type & YYYYXX & XXXXXX \\
        & Parameter Group & 0-1-2-3-4-5 & 0-1-2-3-4-5 \\
        \midrule
        \multirow{2}{*}{$p=52$} & Operator type & XXXXXX & XXXXXX \\
        & Parameter Group & 0-1-2-0-4-4 & 0-1-1-3-4-5 \\
        \midrule
        \multirow{2}{*}{$p=82$} & Operator type & XXXXXX & XXXXXX \\
        & Parameter Group & 0-1-0-0-0-1 & 0-0-0-0-0-0 \\
        \bottomrule[1pt]
    \end{tabular}
    \label{tab:mh}
\end{table*}

 \subsection{Convergence of QAOA with various mixer Hamiltonian}

In our investigation, we conducted an analysis on a randomly selected 16-qubit Max-Cut and TFIM problem from our test dataset, scrutinizing the convergence patterns of QAOA, ma-QAOA, and our method across various configurations ($p=4,6,8,10$). Illustrated in Fig.~\ref{fig:convergence_n16}, our methodology not only achieves a notably lower loss value within a reduced number of iterations in comparison to both QAOA and ma-QAOA but also consistently outperforms in terms of the final loss value attained by the end of the optimization. Specifically, at $p=10$, our approach necessitates merely $28$ iterations for Max-Cut and $22$ iterations for TFIM to diminish the loss value to $-8$ and $-15$, respectively. In contrast, ma-QAOA demands $40$ iterations for both challenges, whereas QAOA fails to achieve this loss value. This evidence underscores the superior efficiency and effectiveness of our method in navigating the solution landscape for these quantum optimization tasks.

\begin{figure}[htp]
    \centering
    \includegraphics[width=1.0\linewidth]{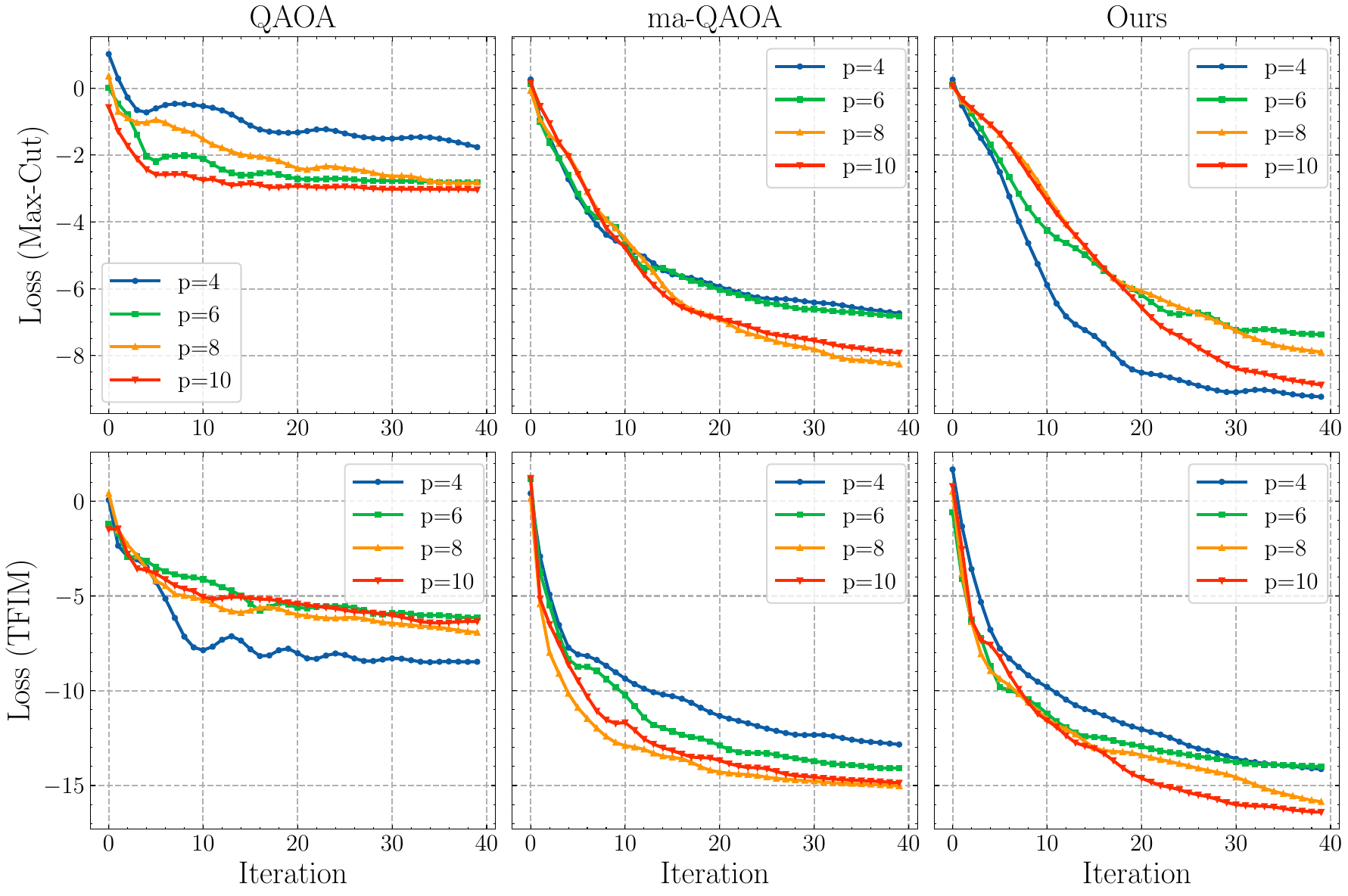}
    \caption{\small{\textbf{Comparison of the convergence of $16$-qubit QAOA, ma-QAOA and our model for Max-Cut and TFIM with varying $p$.} }}
    \label{fig:convergence_n16}
\end{figure}

\subsection{Experiments on extended candidate operator type set}

In this section, we investigate the performance of our model when applied to a more complex set of candidate operator types. Specifically, we expand the pool of mixer operator types from ${X,Y}$ to ${X,Y,XX,YY}$ by incorporating additional two-qubit operators, thereby increasing the search space for operator types to $O(4^N)$. All other experimental conditions remain consistent with those described in the main text. The behavior of the cost estimator under these conditions is illustrated in Fig.~\ref{fig:xx_yy}. Our results indicate that the cost estimator continues to serve as a reliable performance indicator for QAOA, even with the increased complexity of the mixer Hamiltonian design.

\begin{figure}[htp]
    \centering
   \includegraphics[width=0.7\linewidth]{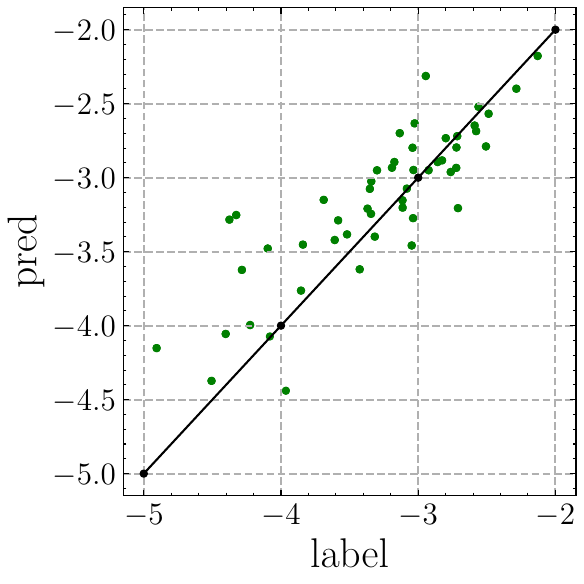}
    \caption{\small{\textbf{Behavior of cost estimator with extended mixer operator pool $\{X,Y,XX,YY\}$.} `label' represents the actual achieved approximation ratio, while `pred' represents the result predicted by the cost estimator.}}
    \label{fig:xx_yy}
\end{figure}

\subsection{Ablation study on the circuit depth embedding}

MG-Net acts as an initial protocol and provides a flexible circuit-generation framework where model components can be conveniently replaced by advanced techniques. Besides the position embedding of circuit depth in the main text, we have also considered another two embedding strategies: integer embedding and one-hot embedding. There are two key differences between the implementation of position encoding and one-hot or integer encoding:
\begin{enumerate}
    \item \textbf{Feature vector length.} The length of the one-hot-encoded vector $\mathbf{x}_p$ depends on the predefined maximum value of $p$, while the length of the integer-encoded vector $\mathbf{x}_p$ is 1. In contrast, we adjust the length of position-encoded vector $\mathbf{x}_p$ according to the dimension of $\mathbf{x}_C$ and $\mathbf{x}_M$.
    \item \textbf{Feature integration strategy.} When using one-hot or integer encoding, we employ concatenation as the integration strategy for the three features $\mathbf{x}_C$, $\mathbf{x}_M$ and $\mathbf{x}_p$ rather than summation.
\end{enumerate}

The achieved approximation ratios for 6-qubit MaxCut problems using different depth encoding methods are shown below:

\begin{table}[htp]
    \centering
    \caption{\small{\textbf{Comparison of approximation ratio $r$ among different circuit depth embedding strategies.}}}
    \begin{tabular}{lc}
        \toprule[1pt]
        Depth embedding method & Approximation ratio $r$  \\
        \midrule
        Integer & $0.981\pm 0.004$ \\
        One-hot & $0.984\pm 0.003$ \\
        Position & $0.99\pm 0.0004$ \\
        \bottomrule[1pt]
    \end{tabular}
    \label{tab:ar_depth_embed}
\end{table}

\end{document}